\documentclass[11pt,letter]{article}

\newcount\Colorflag
\usepackage[margin=1in]{geometry}

\usepackage{enumitem}
\usepackage{amssymb}
\usepackage{natbib}
\usepackage{amsmath,amssymb,amsthm,bm}
\usepackage{mathabx}
\usepackage{hyperref,enumitem}
\setlist{noitemsep,leftmargin=\parindent,topsep=2pt}
\usepackage{xspace}
\usepackage{graphicx}
\usepackage{nicefrac}
\usepackage[font=small]{caption}
\usepackage{subcaption}
\usepackage{multirow}
\usepackage{bbm}
\usepackage[toc,page]{appendix}
\usepackage{wrapfig}
\usepackage{enumitem}
\usepackage{amsmath}
\usepackage{algorithm}
\usepackage[noend,]{algorithmic}
\newcommand{\Gn}{\mathbb{G}_n}
\newcommand{\Ep}{\mathbb{E}}


\newcount\Comments 
\newcommand{\E}{\mathbb{E}}

\newcommand{\G}{\mathbb{G}}

\renewcommand{\Pr}{\ensuremath{\mathrm{Pr}}}

\theoremstyle{plain}
\newtheorem{theorem}{Theorem}[]

\newtheorem{lemma}[theorem]{Lemma}
\newtheorem{definition}[]{Definition}
\newtheorem{condition}[]{Condition}
\newtheorem{remark}[]{Remark}

\newtheorem{assumption}[]{ASSUMPTION}

\newcounter{example}[section]
\newenvironment{example}[1][]{\refstepcounter{example}\par\medskip
   \noindent \textbf{Example~\theexample. #1} \rmfamily}{\medskip}

\newcommand{\ba}{\begin{array}}
\newcommand{\ea}{\end{array}}
\newcommand{\bs}{\begin{align}\begin{split}\nonumber}
\newcommand{\bsnumber}{\begin{align}\begin{split}}
\newcommand{\es}{\end{split}\end{align}}

\renewcommand{\qed}{\hfill{\tiny \ensuremath{\blacksquare} }}%
\renewcommand{\Pr}{{\mathrm{P}}}

\newcommand{\En}{{\mathbb{E}_n}}

\newcommand{\EN}{{\mathbb{E}_{N}}}

\newcommand{\GN}{\mathbb{G}_N}

		{\end{pmatrix}\end{medsize}}%

\renewcommand{\qed}{\hfill {\tiny {\ensuremath{\blacksquare}}}}

\DeclareMathOperator{\eig}{eig}

\theoremstyle{definition}

\renewcommand{\Pr}{{\mathrm{P}}}

\renewcommand{\Pr}{{\mathrm{P}}}

\renewcommand{\leq}{\leqslant}
\renewcommand{\geq}{\geqslant}
\renewcommand{\qed}{\hfill{\tiny \ensuremath{\blacksquare} }}%
\renewcommand{\Pr}{{\mathrm{P}}}

\theoremstyle{definition}

\begin{document}
\title{Machine Learning for Dynamic Discrete Choice}
\author{
	Vira Semenova \thanks{I am deeply grateful to my advisors Victor Chernozhukov, Whitney Newey, and Anna Mikusheva for their guidance and encouragement. I am thankful to Alberto Abadie,  Chris Ackerman, Sydnee Caldwell, Denis Chetverikov, Ben Deaner, Mert Demirer, Jerry Hausman, Peter Hull, Tetsuya Kaji,  Kevin Li, Elena Manresa, Rachael Meager, Denis Nekipelov,  Cory Smith,  Sophie Sun,  Roman Zarate for helpful comments.  }
}

\date{\today}
\maketitle

 \begin{abstract}
  
Dynamic discrete choice models often discretize the state vector and restrict its dimension in order to achieve valid inference. I propose a novel two-stage estimator for the set-identified structural parameter that incorporates a high-dimensional state space into the dynamic model of imperfect competition. In the first stage, I estimate the state variable's law of motion and the equilibrium policy function using machine learning tools. In the second stage, I plug the first-stage estimates into a moment inequality and solve for the structural parameter. The moment function is presented as the sum of two components, where the first one expresses the equilibrium assumption and the second one is a bias correction term that makes the sum insensitive (i.e., Neyman-orthogonal) to first-stage bias. The proposed estimator uniformly converges at the root-$N$ rate and I use it to construct confidence regions.  The results developed here can be used to incorporate high-dimensional state space into classic dynamic discrete choice models, for example, those considered in \cite{Rust}, \cite{BBL}, and \cite{Scott:2013}.
 
\end{abstract}
\section{Introduction}

In empirical work on dynamic models, economists often make specification choices - for example, discretize state space or select covariates for flow utility - in order to achieve computational tractability and precise estimates (e.g. \cite{PakesShankman:1984}, \cite{Pakes:1986}, \cite{Rust}, \cite{Ryan:2012}). Typically, economists have little intuition about which covariates to select or how to discretize a continuous state variable (\cite{PakesLanjouw:1998}). Unfortunately,  counterfactual predictions of dynamic models are sensitive to specification choices and  are difficult to interpret when a model is misspecified. As discussed in \cite{athey2017science}, there has been recent interest in data-driven model selection based on modern machine learning tools. Moreover, (\cite{orthogStructural}, \cite{doubleml2016}, \cite{LRSP}, \cite{AtheyWager}) have shown how to leverage these tools into high-quality estimates of causal parameters. Because dynamic models are more challenging to analyze, the model selection in dynamic models   remains an open question.

This paper estimates a dynamic model of imperfect competition with a  high-dimensional state space.  Consider the bus engine replacement model from \cite{Rust} as a special case. An agent decides whether to replace the bus engine  in each period.  The agent incurs a  fixed cost in the case of replacement  and a cost  proportional to the current mileage in the case of maintenance.  One can imagine that the future bus mileage depends on a vector of  observed  exogenous characteristics, such as current traffic and weather conditions, encoded in a high-dimensional vector. While these characteristics do not enter into the per-period utility of the owner, they affect the mileage's law of motion in the case of maintenance decision, and hence are taken into account in the agent's optimal renewal policy. Therefore, the expected value of  bus ownership depends on a high-dimensional state  vector consisting of the current mileage and exogenous characteristics.   I am interested in the identified set of  the  possible values of the cost parameters that are rationalized by the agent's optimal behavior. In addition to the model in \cite{Rust},  the methods developed here apply to  a  broad variety of dynamic models: for example, those considered in \cite{BBL} or \cite{Scott:2013}, etc. that point$-$ or partially identify their parameters under various  assumptions about the agent's optimal behavior. 

The main difficulty of this approach is estimating the value function.  The plug-in (naive) approach of \cite{BBL} consists of two steps. In the first step, one estimates the state variable's  law of motion and the equilibrium policy function,  effectively recovering agents' equilibrium beliefs. In the second step,  one estimates the value function  by drawing a sequence of states and actions from  the  conditional distributions estimated in the first  step and averaging over multiple simulation draws. If the state variable is high-dimensional, we must regularize the  estimate of the first-stage parameter  in order to achieve consistency in high dimension. An inherent cost of these methods is  bias that converges slower than parametric rate. As a result, first-stage bias carries over into the second stage, resulting in  a low-quality estimate of the identified set.

The major challenge of this paper is to overcome transmitting the bias from the first to the second stage. A basic idea, proposed in a point-identified case, is to make the moment equation insensitive, or, formally, Neyman-orthogonal, to the biased estimation of the first-stage parameter (\cite{Neyman:1959},\cite{doubleml2016}).  The second idea consists in the use  of different samples to estimate the equilibrium beliefs at the first stage and to compute the sample average of the bias correction term at  the second stage. Using different samples in the first and the second stages allows me to employ modern machine learning methods to estimate the first stage parameters.

The first contribution of this paper is to  derive a Neyman-orthogonal moment  for the value function.  This equation is presented as the sum of two terms, where the first  term is the value function itself and the second term is a  bias correction term that makes the sum insensitive (i.e., Neyman-orthogonal) to the first-stage bias. To derive the bias correction term,
 I  characterize the value function as a solution to the recursive (Bellman) equation, which equates the expected value at the current state to the sum of expected immediate payoff and the  expected discounted future payoff.  This equation can be viewed as a semiparametric moment equation, where the parametric component gives the  value function's moment (i.e., weighted average)  and   the  nuisance parameter consists of the first-stage parameters appearing in the value function (i.e., conditional choice probability and state variable's law of motion). The first-stage parameter appears both inside and outside the value function. For example,   the conditional probability of a given choice appears as a weight on immediate and discounted future payoffs corresponding to that choice.  Applying the implicit function theorem to this equation, I derive  the bias correction term and show how to approximate it by simulation. This derivation is novel: this is the first result in the literature that derives the bias correction term for a moment function that is not available in the closed form.

The  second contribution of this paper is to extend the general theory of  moment inequalities proposed by \cite{CHT} to allow for moment functions that depend on a first-stage nuisance parameter that can be high-dimensional (e.g., conditional choice probability). I characterize the identified set as the minimizer of the criterion function that penalizes the incorrect sign of the moment inequality (e.g., the sign that contradicts the optimality assumption). I show that, if the moment function is insensitive with respect to the biased estimation of its nuisance parameter  at each point of the space of the structural parameter, plugging in the first-stage estimate of the nuisance parameter into the moment function delivers a high-quality sample criterion function. In particular, the estimator of the identified set obtained by  inverting  the  sample  criterion function converges at the same rate as if the true value of the first-stage nuisance parameter were known. Furthermore,  inferential statistics based on the estimated moment function have a non-degenerate large sample distribution and are used to construct confidence regions for the identified set by subsampling.

This paper leaves a number of open questions. First, this article assumes that the first-stage nuisance parameter is identified. This assumption does not hold for every application (e.g., \cite{CilibertoTamer}). Second,  the high-dimensional state space introduces a wealth of feasible suboptimal Markov policies to choose from for the construction of moment inequalities. In this paper I take the set of chosen alternatives as given, leaving the optimal choice of these alternatives for future research.

The structure of the paper is as follows.  Section \ref{sec2:setup} gives a brief overview of the results.
Section \ref{sec2:dynamic} gives the low-level sufficient conditions for the dynamic discrete choice model in \cite{BBL}. Section \ref{sec2:theory} presents an asymptotic theory for the identified sets defined by the  semiparametric moment inequalities.

%
%The  first contribution of this paper is the construction of a high-quality estimator of the value function in the presence of the high-dimensional state vector.  To overcome the translation of the  bias discussed above, I use two different ideas the construction of the insensitive  (i.e, Neyman-orthogonal in \cite{doubleml2016}) moment equations and sample splitting.
%The insensitive estimator of the value function is constructed by adding a mean zero bias correction term to the original  estimator of the value function. The resulting moment function is insensitive to the biased estimation of the estimated first stage.  The second idea consists in the use  of different samples to estimate the equilibrium beliefs at the first stage and to compute the sample average of the bias correction term at  the second stage. The use of different samples between the first and the second stages allows me to employ modern machine learning methods to estimate the first stage parameters.
%

\subsection{Literature review}

This paper is built on three lines of work: estimation and inference in partially identified models and Neyman-orthogonal semiparametric estimation.

The first line of research,  see e.g.   \cite{Rosen:2006}, \cite{CHT}, \cite{Romano:2006} develops a framework for the estimation and inference of  identified sets that are partially identified by moment  inequalities.  Extending this framework, \cite{Kaido:2014} allow the moment function to depend on an identified low-dimensional first-stage parameter in addition to the target. Furthermore, I allow the first-stage parameter to be a  high-dimensional vector or a highly complex function and estimate it by modern machine learning methods.

Within the first line of my research, my application is most connected to the estimation of dynamic models of imperfect competition. Specifically, I build on  \cite{BBL} who identifies the structural parameter as a solution to a set of moment inequalities that embody the  assumptions about the agent's optimal behavior. This paper proposed a two-stage algorithm to estimate the parameter, where in the first-stage one estimates the state variables's law of motion and equilibrium policy function, and then plugs them into a moment equation derived from the equilibrium assumption. However, to achieve valid inference,  \cite{BBL} imposed parametric restrictions in the first stage. Extending this algorithm, I drop these restrictions and estimate first-stage parameters by machine learning methods.

The second line of research (\cite{Hasminskii:1981},\cite{Andrews:1994}, \cite{Newey1994}, \cite{vdv}) is concerned with obtaining a root-$N$ consistent and asymptotically normal estimate for a low-dimensional target parameter in  the  presence of a nuisance parameter. In this literature, a two-stage statistical procedure is  insensitive, or, formally, Neyman-orthogonal,  to the estimation error of the first-stage parameter (\cite{Neyman:1959}). Extending the orthogonality idea from a parametric to semiparametric setup was done in \cite{Newey1994}, \cite{RRZ},  \cite{Robins}. Combining  Neyman-orthogonality and  sample splitting, \cite{doubleml2016} and  \cite{LRSP} incorporated modern machine learning methods to estimate low-dimensional target parameters defined by semiparametric moment equations. Subsequently, the idea has been extended to the case of high-dimensional target parameter in \cite{CGST} and  \cite{DenisVas}. This paper translates the idea of Neyman-orthogonality from point- to set-identified case.

Within the second line research, my application is most connected to the estimation of dynamic discrete choice models under point-identification (\cite{BHKN}, \cite{BCHN}, \cite{Arcidiacono:2013},  \cite{LRSP}). Specifically,  \cite{LRSP} introduces high-dimensional state space into a  dynamic discrete choice model whose choice set contains a renewal choice and derives a Neyman-orthogonal moment equation for the structural parameter in that model.  In this paper, I address the cases that do not have renewal choice property and derive  the Neyman-orthogonal moment equation for the value function directly.  This result is applicable to both point- and set-identified cases with discrete and continuous choice sets.

%derives  a } 

%,, estimate structural parameters in static and dynamic models.This paper translates the Neyman-orthogonality idea from point$-$ to set-identified case. Other app

\section{Set-Up and Motivation}
\label{sec2:setup}

I am interested in an identified set defined by moment inequalities
\begin{align}
\label{eq:theta0}
    \Theta_I:= \{ \theta: \E m(D,\theta,\eta_0) \leq 0\},
\end{align}
where $D$ is the data vector distributed as $P_D$, $\theta$ is the parameter of interest, and $\eta$ is an identified yet unknown parameter of the distribution $P_D$ whose true value is $\eta_0$.  For example, in the bus engine replacement model of  (\cite{Rust}) the parameter of interest, $\theta$, is a vector of operational and replacement costs, the data vector $D$ contains  bus mileage and other observed bus characteristics,  $\eta$ contains mileage's law of motion and the conditional probabilities of bus replacement, and the inequality restriction comes from the assumption  that agent behaves optimally.  I allow the state variable, $w$ , to be high dimensional and estimate the parameter $\eta$ by modern machine learning methods.

The examples below demonstrate how a high-dimensional state $w$ may appear in the dynamic discrete choice model.

\begin{example}[Engine Replacement Model from \cite{Rust} with a High-Dimensional State Variable]
\label{ex:trans}
A single agent  makes a binary decision $a \in \mathcal{A} = \{0,1\}$ whether to replace a bus engine in each period $t \in \{1,2,\dots, \infty\}$. His per-period utility function is
\begin{align}
\label{eq:utper}
    \pi(a,s,\epsilon) = \begin{cases} -R +  \epsilon(0),\quad a=0,\\
    -\mu \cdot s  + \epsilon(1),\quad a=1, \end{cases}
\end{align}
where $s \in \mathcal{R}$ is the bus mileage, $R$ is replacement costs, $\mu s$ is operational cost,  $a$ is the decision of the agent, and $\epsilon = (\epsilon(0),\epsilon(1))$ is a vector of private shocks associated with each decision. The state variable $w$ consists of the mileage $s$ and  additional high-dimensional vector of exogenous variables $x$ (e.g., engine manufacturer characteristics) that I assume  do not change with time. After the replacement decision $(a=0)$ the mileage $s_{next}$ resets to $1$ with probability one. After  the maintenance decision $(a=1)$,  the mileage $s_{next}$ follows a first-order Markov process
\begin{align}
\label{eq:trans:setup}
    s_{next} &= \rho_0(w) + e, \quad e \sim N(0,1),
\end{align}
where $\rho_0(w)$ is the conditional expectation function of the mileage tomorrow $s_{next}$ given the state today $w$ and $e \sim N(0,1)$ is an independent $N(0,1)$ shock.  The target parameter $\theta = (R,\mu)$ consists of the replacement and operational cost parameters. The observed data vector consists of the current state, action, and future state (i.e, $D=(w,a,w_{next})$). Under the assumptions discussed below, the unknown yet identified high-dimensional parameters consist of the conditional choice  probability $\gamma(w)=\Pr(a=1|w)$ and the transition function $\rho(w)$.

\end{example}

%\begin{footnote}{To see why the agent takes into account the characteristics $x$, consider a red bus and a blue bus. The blue bus is a normal bus whose mileage increments by one in the case its engine is not replaced. The mileage of red bus explodes in the case its engine is not replaced for any mileage value $s \geq 1$. Therefore, the optimal agent's decision cannot be based solely on the bus mileage and should take into account the colour as well. } \end{footnote}

\begin{example}[Entry Game with a Long-Lived and a Short-Lived Player]
\label{ex:Pug}

In each period $t \in \{1,2,\dots, \infty\}$ Apple decides whether to issue a new model of a phone $(a_{Apple}=0)$ or keep the existing one $(a_{Apple}=1)$. In addition to Apple, a short-lived potential entrant (Player 2) decides whether to issue a fake $(a_{P}=1)$ or not $(a_{P}=0)$. Each period Apple faces a new copy of player 2.  Player 2's  actions do not influence the motion of the state, and he has no dynamic incentives. In each period both players observe a state vector $w = (s,x)$ that consists of the age of the current make $s$ and the vector $x$ of short-lived player's characteristics (e.g., country, information about intellectual property rights protection) that does not change with time (i.e, $x_t=x_0$).

In each period both players receive a privately observed shock. Apple's utility function is given by
\begin{align}
\label{eq:utpug}
    \pi(a,w,\epsilon) = \begin{cases}
    - R + \delta_1 a_{P} + \epsilon(0), \quad a_{Apple}=0,\\
    -\mu \cdot s + \delta_2 a_{P} + \epsilon(1), \quad a_{Apple}=1,
    \end{cases}
\end{align}
where $R$ is the fixed cost of replacing the current  model with a new one, $-\mu \cdot s$ is the profit from the current   make that decays with age, and $\epsilon = (\epsilon(0),\epsilon(1))$ is a vector of Apple's shocks associated with each decision. After the replacement decision $(a_{Apple}=1)$ the age $s_{next}$ resets to $1$ with probability one. After the maintenance decision $(a_{Apple}=2)$ the age $s_{next}$ increases by $1$ with probability $1$
$$s_{next} = s +1.$$

The target parameter $\theta = (R,\mu,\delta_1-\delta_0)$ consists of the cost parameters and the difference between the interaction parameters $\delta_1-\delta_0$. The observed data vector $D=(w,a,w_{next})$ consists of the current state, $w$, action profile $a = (a_{Apple},a_{P})$, and the future state, $w_{next}$. Under the assumptions discussed below, the unknown yet identified high-dimensional parameters consist of the conditional choice  probability of both players: $\gamma_A(w)=\Pr(a_{Apple}=1|w)$  and $\gamma_P(w):=\Pr(a_{P}=1|w)$.

\end{example}

Consider the setting of Example \ref{ex:trans}. I assume that the  agent follows a Markov policy $\sigma(w,\epsilon)$ that maps the current state $w \in \mathcal{W} \subset \mathcal{R}^{d_w}$ and the shock vector $\epsilon \in \mathcal{R}^2$ into the action space $\mathcal{A} = \{0,1\}$. The value function $  V(w;\theta;\sigma)$ of a Markov policy $\sigma$ is given by
\begin{align*}
     V(w;\theta;\sigma) = \E [\sum_{t =0}^{\infty} \beta^t \pi(\sigma(w_t,\epsilon_t),w_t,\epsilon_t)|w],
\end{align*}
where $\beta <1$ is a discount factor. This function  can also be written recursively:
\begin{align}
\label{eq:recursive}
     V(w;\theta;\sigma) = \E_{\epsilon} \big[  \pi(\sigma(w,\epsilon),w,\epsilon(\sigma(w,\epsilon))) + \beta \E [ V(w_{next};\theta;\sigma )|w,\sigma(w,\epsilon)] \big],
\end{align}
where the  first and  second summands show the expected current profit and the future expected discounted value, respectively.   Define the choice-specific value function as
\begin{align*}
    v (a,w)&:= -R_0 (1-a) + (-\mu_0 s ) a  + \beta \E [V(w_{next};\theta_0;\sigma^{*})|w,a=0](1-a)\\
     &+ \beta \E [V(w_{next};\theta_0;\sigma^{*})|w,a=1]  a, \quad a \in \mathcal{A},
\end{align*}
where the current deterministic utility is evaluated at the action $a$ and the future discounted value is evaluated at the optimal strategy $\sigma^{*}$, conditional on the current action $a$.  The symbols $ (R_0,\mu_0)$ stand for the true values of the cost parameters. The choice $a$ is optimal if and only if its total utility is greater than  or equal to the utility of any other choice $a' \in \mathcal{A}$:
\begin{align*}
     v (a,w) + \epsilon(a) \geq v (a',w)+ \epsilon(a'), a' \in \mathcal{A}.
\end{align*}
Then the optimal Markov policy $\sigma^{*}(w,\epsilon)$ has a cutoff form:
\begin{align}
\label{eq:opt:sigma}
    \sigma^{*}(w,\epsilon) &= \begin{cases}1, \quad \epsilon(1) - \epsilon(0) \geq -v(1,w) +v(0,w)  \\
    0, \quad \text{otherwise}.
    \end{cases}
\end{align}
To identify the difference
$(v(1,w) - v(0,w))$ I  make the following standard assumption (e.g., \cite{HotzMiller}, \cite{Scott:2013}) that I will maintain throughout the paper. In general case, the model will involve several players $K$. Denote their choice sets by $\mathcal{A}_k, k\in \{1,2,\dots,K\}$.
\begin{assumption}[Independent logit errors]
\label{ass:logit}
The components of the private shock vector
$\epsilon = \{ (\epsilon_k(j))_{j \in \mathcal{A}_k}\}_{k=1}^K \in \mathcal{E}$
  are identically and independently distributed with a type $1$ extreme value distribution whose distribution  function is equal to $F(t) =  \exp(-\exp(-t))$.
\end{assumption}
As discussed in \cite{HotzMiller}, the vector of differences of the choice-specific value functions can be expressed as
\begin{align}
\label{eq:HotzMiller:setup}
    v(a_1,w) - v(1,w) &= \log \frac{\Pr (a_1|w)}{\Pr (1|w)} = \log \frac{\gamma_0(w)}{1-\gamma_0(w)},
\end{align}
where $\gamma_0(w)=\Pr (a=1|w)$ is the probability of the decision to maintain the engine conditional on the state $w$ which is identified. For expositional purpose, I  consider a simple suboptimal Markov policy: the choice of the decision based on the coin toss:
\begin{align}
\label{eq:cointoss}
    \E (\sigma(w,\epsilon) = 1|w) =\frac{1}{2}.
\end{align}

Combining \ref{eq:opt:sigma} and \ref{eq:HotzMiller:setup}, I  recognize that the optimal strategy can be viewed as a function of $\gamma$:
$$\sigma^{*}(w,\epsilon) = \sigma^{*}(w,\epsilon,\gamma)$$ Therefore, value function $  V(w;\theta;\sigma^{*})$  
can be viewed as the function of $\gamma$.

%. We will abstract away from the estimation issue of 

%In what follows, the value function $V(w;\theta;\sigma)$ is estimated in two stages. In the first stage, 
%o estimate the optimal policy 
%$ \sigma^{*}(w,\epsilon)$, I will plug in an estimate $\hat{\gamma}(w)$ into (\ref{eq:HotzMiller:setup}). 
%In what follows, I indicate the dependence of the value function $  V(w;\theta;\sigma)$ and the optimal strategy $\sigma^{*}(w,\epsilon)$ on the nuisance parameter $\eta$:

Define the moment function $m(w,\theta,\gamma)$ as the difference of the value function evaluated for the suboptimal and the optimal strategies
\begin{align}
\label{eq:ineq:trans}
   m(w,\theta,\gamma):= V(w;\theta;\sigma;\gamma) -V(w;\theta;\sigma^{*};\gamma).
\end{align}
Define the identified set   $\Theta_I$ as
\begin{align*}
    \Theta_I:= \E [m(w,\theta,\gamma_0)  ] \leq 0.
\end{align*}
Because $\sigma^{*}$ is an optimal strategy, the inequality above holds for the true parameter $\theta_0$, and $\Theta_I$ is a valid identified set.

\subsection{Naive Approach to the Estimation of the Identified Set}

The function $m(w,\theta,\eta)$ given in (\ref{eq:ineq:trans}) presents two complications. First, the value function $V(w;\theta;\sigma;\eta)$ depends on the unknown  nuisance parameter $\eta$. Second, even if the value of $\eta$ is given, the value function is not readily available in the closed form and must be approximated by simulation. Here I focus on the first complication as if the moment function were readily available, leaving the description of the simulation algorithm for Section \ref{sec2:dynamic}.

A possible, though naive estimator of the value function can be constructed as follows. Consider an ideal scenario where the researcher knows mileage's law of motion.
Then  the choice probability $\gamma$ is the only unknown nuisance parameter that appears in the value function.  Given an i.i.d sample $(D_i)_{i=1}^N$ from the law $P_D$, it is split  into a main sample $J_2$ and an auxiliary sample $J_1$ of equal size $n =[N/2]$ such that  $J_1 \sqcup J_2=\{1,2,\dots,N\}$. After that, the estimator of the value function $V(w;\theta;\sigma^{*};\gamma) $ is constructed as the sample average:
\begin{align*}
%    \label{eq:naive}
    \widehat{V}(w;\theta;\sigma^{*};\widehat{\gamma}):= \frac{1}{|J_2|} \sum_{i \in J_2} V(w_i;\theta;\sigma^{*};\widehat{\gamma}(w_i)),
\end{align*}
where $\widehat{\gamma}$ is estimated on the auxiliary sample $J_1$. Unfortunately, for some parameter values $\theta \in \Theta$ this estimator has slower than $\sqrt{N}$ convergence:
\begin{align}
\label{eq:slow:chap2}
    \sqrt{N}  |\widehat{V}(w;\theta;\sigma^{*};\widehat{\gamma})  - \E V(w;\theta;\sigma^{*};\gamma_0) | \rightarrow \infty.
\end{align}
Therefore, the estimator $\widehat{\Theta}_I$ of the identified set $\Theta_I$ based on the moment function (\ref{eq:ineq:trans})  has suboptimal convergence rates.

The source of the slow convergence (\ref{eq:slow:chap2}) can be understood through the following decomposition:
\begin{align*}
      \sqrt{N}  (\widehat{V}(w;\theta;\sigma^{*};\widehat{\gamma}) - \widehat{V}(w;\theta;\sigma^{*};\gamma_0) )   &=   \underbrace{\sqrt{N}  ( \E [V(w_i;\theta;\sigma^{*};\gamma_0)] - \frac{1}{|J_2|} \sum_{i \in J_2} V (w_i;\theta;\sigma^{*};\gamma_0)  )}_{\mathbf{a}}\\
      &+ \underbrace{\sqrt{N} \E [V(w_i;\theta;\sigma^{*};\widehat{\gamma}) -V(w_i;\theta;\sigma^{*};\gamma_0)]}_{\mathbf{b}} \\
      &+ \underbrace{\sqrt{N}  (  \frac{1}{|J_2|} \sum_{i \in J_2} V(w_i;\theta;\sigma^{*};\widehat{\gamma})  -  \E V(w_i;\theta;\sigma^{*};\widehat{\gamma})  ) }_{\mathbf{c}}.
      \end{align*}
      The term $\mathbf{a}$ is the centered sample average of the value function $V(w;\theta;\sigma^{*};\gamma_0) $ evaluated at the true value $\gamma_0$ of the choice probability. Due to the sample splitting, the term $\mathbf{c}$ is a centered sample average of conditional on $J_1$ i.i.d random variables and is well-behaved.      The term $\mathbf{b}$ stands for the bias of the value function $V(w;\theta;\sigma^{*}; \widehat{\gamma})  $  coming from the biased estimation of the conditional choice probability $\gamma$. This term is responsible for the slow convergence (\ref{eq:slow:chap2}). 

The divergence of  $\mathbf{b}$ comes from the combination of two facts: biased estimation of the first stage and the transmission of the bias from the first to the second stage. Since the conditional choice probability $\gamma(w)$ is a function of a high-dimensional state vector $w$, bias of its machine learning estimate (e.g., $\ell_1$-regularized logistic regression) 
converges slower than root-$N$.  Because the value function $V(w;\theta;\sigma^{*};\gamma)  $ is sensitive to this bias, it carries over into the second stage.

 \paragraph{Overcoming the Regularization Bias using Orthogonalization.}

 To overcome the translation of the first stage bias  into the second stage I add the bias correction term  to the value function $V(w;\theta;\sigma^{*};\gamma)  $ to make it insensitive with respect to the biased estimation of $\gamma$. The new moment function for the value function $V(w;\theta;\sigma^{*};\gamma)$ takes the form
 \begin{align*}
     g(D;\theta;\sigma^{*};\gamma):= V(w;\theta;\sigma^{*};\gamma) + \frac{1}{1-\beta} \Gamma(w;\theta) (1_{a=1} - \gamma(w)),
 \end{align*}
 where the function $\Gamma(w;\theta)$ for $\theta = (R,\mu)$ is 
 \begin{align}
 \label{eq:big:gamma}
    \Gamma(w;\theta)&:= - \mu s + R + \beta \E  [ V(w_{next};\theta;\sigma^{*};\gamma_0)|w,a=1] \\
    &- \beta \E  [ V(w_{next};\theta;\sigma^{*};\gamma_0)|w,a=0]  - \frac{2}{\gamma(w)} - \log \frac{\gamma(w)}{1-\gamma(w)}. \nonumber
 \end{align}
Because the bias correction term has zero mean
 \begin{align*}
     \E [\frac{1}{1-\beta} \Gamma(w;\theta)   (1_{a=1} - \gamma_0(w))] &=\frac{1}{1-\beta} \E_{w} \Gamma(w;\theta)   \E [(1_{a=1} - \gamma_0(w))|w ] =0, \quad \theta \in \Theta,
 \end{align*}
 the new moment function can replace the old one in the definition of the identified set. Moreover, the new function is  insensitive to the bias estimation of $\gamma$.  As a result, under additional mild regularity conditions,  the regularization bias of the choice probability $\widehat{\gamma}(w)$ does not translate into the bias of the moment function:
 \begin{align*}
    \sqrt{N} | \E g(D,\theta,\widehat{\gamma}) - g(D,\theta,\gamma_0)] |\rightarrow 0.
 \end{align*}

\paragraph{The Role of Sample Splitting in Preventing Overfitting.}
Another key aspect of the proposed analysis is using different samples $J_1$ and $J_2$ for different stages of  estimating  the value function. Had I used the whole sample to estimate the conditional choice probability $\widehat{\gamma}(w)$, the sample average
\begin{align*}
    \frac{1}{N} \sum_{i=1}^N g(D_i, \theta, \widehat{\gamma})  - \E [g(D_i, \theta, \widehat{\gamma})  ]
\end{align*}
would not be a sample average of the i.i.d (or weakly dependent) random observations. The relation between the first stage error $\widehat{\gamma}(w_i) - \gamma_0(w_i)$ and the value of the moment function $ g(D_i, \theta, \gamma_0)$ creates bias, referred to as overfitting bias. To control the overfitting bias in the worst-case scenario:
\begin{align*}
    \E \sup_{\gamma \in \mathcal{G}} |\sum_{i=1}^N (g(D_i, \theta, \gamma)  - g(D_i, \theta, \gamma_0)) - \E [g(D_i, \theta, \gamma) -g(D_i, \theta, \gamma_0) ]|,
\end{align*}
one must impose complexity constraints on the class of functions $\mathcal{G}$ that are used to estimate the conditional choice probability. While some nonparametric estimators designed  for low-dimensional state spaces obey these constraints, some  modern machine learning  estimators designed for high-dimensional state variables do not.  To accommodate machine learning estimators at the first stage, I use different samples. 

\paragraph{Bias Correction Terms for  Examples \ref{ex:trans} and \ref{ex:Pug}}. Suppose a nuisance parameter can be presented as a conditional expectation $\E[U|w]$.  Then bias correction term for conditional choice probability takes the form
 $$ \alpha(D;\theta;\xi) = \Pi (w;\theta) ( U - \E [U|w]),$$
where $\xi$ is an unknown vector-valued function of the state variable $w$. The true value $\xi_0 = \xi_0(\theta)$ consists of the original conditional expectation function $\E [U|w]$ and the function $ \Pi (w;\theta) $
\begin{align}
\label{eq:xi0trans:chap2}
    \xi_0(\theta):= \{ \E [U|w], \Pi(\cdot,\theta) \}.
\end{align}
%The function $\Pi (w;\theta)$ can be thought of as the effect of the estimation error $\gamma(w) - \gamma_0(w)$ of the expectation of the moment function $m(w,\theta,\eta)$ holding the instances of the other nuisance parameters fixed.

\begin{remark}[Example \ref{ex:trans}, continued]
\label{rm:trans}
Consider the setup in Example \ref{ex:trans}. The nuisance parameter $\eta = (\gamma,\rho)$ consists of the conditional choice probability $\gamma$ and the transition function $\rho$ defined in (\ref{eq:trans:setup}). The transition function $\rho$ is present in the value function evaluated for the optimal $\sigma^{*}$ and the suboptimal $\sigma$ strategies. The conditional choice probability enters  $V(w;\theta;\sigma^{*};\eta)$ only through the optimal strategy $\sigma^{*}$, described in (\ref{eq:opt:sigma})-(\ref{eq:HotzMiller:setup}). To sum up, the bias correction term for Example \ref{ex:trans} is 
\begin{align}
\label{eq:alpha:rust}
	\alpha(D;\theta;\xi) &= \alpha^{TRANS}_{\sigma}(D;\theta;\xi) - \alpha^{TRANS}_{\sigma^{*}}(D;\theta;\xi) -  \alpha^{CCP}(D;\theta;\gamma),
\end{align}
where e.g., $\alpha^{TRANS}_{\sigma}(D;\theta;\xi)$ is the individual bias correction term for $\rho$ for the case of suboptimal strategy $\sigma$.

As discussed above, the bias correction term $\alpha^{CCP}(D;\theta;\gamma)$ that corrects the bias of the conditional choice probability $\gamma_0$ is 
\begin{align}
\label{eq:alphaccp}
    \alpha^{CCP}(D;\theta;\gamma) &= \frac{1}{1-\beta} \Gamma(w;\theta)  (1_{\{ a=1\}} - \gamma(w)),
\end{align}
where the function $\Gamma(w,\theta)$ is given in (\ref{eq:big:gamma}). The form for the bias correction term for the transition function $\rho$ is the same regardless  whether $\sigma$ is optimal or not.  Let $\tilde{\sigma} \in \{ \sigma, \sigma^{*}\}$ be a Markov policy.  Suppose the state vector $w$ has a stationary distribution.  Then the bias correction term of the value function $V(w;\theta;\sigma;\eta)$ for the transition function $\rho(\cdot)$ is equal to:
\begin{align}
\label{eq:alpharho}
    \alpha^{TRANS}_{\tilde{\sigma}}(D;\theta;\xi):= \frac{\beta}{1-\beta} \E [\frac{d V(x;\theta;\tilde{\sigma};\eta_0)}{dx}|_{x = w_{next}}|w,a=1 ] ( s_{next} - \rho(w)),
\end{align}
where $\xi$ is an unknown vector-valued function of the state variable $w$. Its true value $\xi_0 = \xi_0(\theta)$ consists of the original nuisance parameter $\eta_0$ and the function $\Pi_0(w,\theta)$:
$$\Pi_0(w,\theta):= \E [\frac{d V(x;\theta;\sigma;\eta_0)}{dx}|_{x = w_{next}} |w ],$$
which is equal to the expectation of the derivative of the value function with respect to the state variable evaluated for the future state conditional on the current state $w$. 

The bias correction term for other suboptimal Markov policies is more complicated, but can be derived using the argument of Appendix \ref{appendix:a}.
%For  the bias correction 

\end{remark}

\begin{remark}[Example \ref{ex:Pug}, continued]
\label{rm:Pug}
Consider the setup in Example \ref{ex:Pug}. The nuisance parameter $\eta = (\gamma_A, \gamma_P)$ consists of the conditional choice probabilities of Apple and Player 2.  The choice probability $\gamma_A$ enters the value function $V(w;\theta;\sigma^{*};\eta)$ only through the optimal strategy $\sigma^{*}$, described in (\ref{eq:opt:sigma})-(\ref{eq:HotzMiller:setup}). The choice probability  $\gamma_P$ enters the value function  $V(w;\theta;\tilde{\sigma};\eta), \quad \tilde{\sigma} \in \{ \sigma^{*}, \sigma \}$. To sum up, the bias correction term for Example \ref{ex:trans} is equal to:
\begin{align*}
    \alpha(D;\theta;\eta):=  - \alpha^{CCP}_{A}(D;\theta;\gamma) + \alpha^{\sigma}_{P}(D;\theta;\gamma_{P}) -  \alpha^{\sigma^{*}}_{P}(D;\theta;\gamma_{P}),
  \end{align*}
where e.g. $ \alpha^{CCP}_{A}(D;\theta;\gamma) $ is the individual bias correction term for  $\gamma_A$.

Let $\tilde{\sigma} \in \{ \sigma, \sigma^{*}\}$ be a Markov policy.  Suppose the state vector $w$ has a stationary distribution.  Then the bias correction term of the value function $V(w;\theta;\sigma;\eta)$ for the $\gamma_P$ is equal to:
\begin{align*}
    \alpha^{\sigma}_{P}(D;\theta;\gamma_{P}):= \Gamma^{\tilde{\sigma}}_{P}(w;\theta) (1_{\{ a_{P}=1\}}-\gamma_{P}(w)),
\end{align*}
where
\begin{align*}
    \Gamma^{\tilde{\sigma}}_P(w;\theta) = \frac{1}{1-\beta}( \delta_0 + (\delta_1 - \delta_0) \gamma^{\tilde{\sigma}}(w)),
\end{align*}
where $\gamma^{\tilde{\sigma}}(w) = \Pr (\tilde{\sigma}(w,\epsilon) = 1|w)$ is the conditional probability of Apple's decision under the policy $\tilde{\sigma}$.
The bias correction term for  Apple's conditional choice probability $\gamma_A$ is
\begin{align*}
    \alpha^{CCP}_A(D;\theta;\gamma) &= - \mu s + R + (\delta_1 - \delta_0) \gamma_A(w) + \beta \E [V (w_{next};\theta;\sigma^{*};\eta_0) |w,a=1]\\
    &- \beta \E [V (w_{next};\theta;\sigma^{*};\eta_0) |w,a=0].
\end{align*}
To sum up, the bias correction term $\alpha(D;\theta;\eta)$ is equal to:
\begin{align*}
	  \alpha(D;\theta;\eta):=  - \alpha^{CCP}_{A}(D;\theta;\gamma) +   \frac{1}{1-\beta} (\delta_1 - \delta_0) (\gamma_{0}^{\sigma}(w) - \gamma_{A,0}(w)) (1_{\{a_{P} =1\}} -\gamma_P(w)).
\end{align*}

\end{remark}
In some point-identified problems the value function $V(w;\theta;\sigma;\eta)$ is only evaluated at the optimal strategy $\sigma^{*}$ and the true  parameter value $\theta_0$.  Then the bias correction term is evaluated only at the true value $\theta_0$.   In particular, in the Examples \ref{ex:trans} and \ref{ex:Pug} the function $\Gamma(w,\theta_0)$ can be further simplified  as follows:
\begin{align*}
\Gamma(w,\theta_0) &= \underbrace{ - \mu_0 s + R_0 + \beta \E  [ V(w_{next};\theta_0;\sigma^{*};\gamma_0)|w,a=1] - \beta \E  [ V(w_{next};\theta_0;\sigma^{*};\gamma_0)|w,a=0] }_{v(1,w) - v(0,w)}\\
&- \frac{2}{\gamma(w)} - \log \frac{\gamma(w)}{1-\gamma(w)} \\
    &=  \log \frac{\gamma(w)}{1-\gamma(w)} + ( - \frac{2}{\gamma(w)} - \log \frac{\gamma(w)}{1-\gamma(w)} ) \\
    &=  - \frac{2}{\gamma(w)},
\end{align*}
where we have used (\ref{eq:HotzMiller:setup}) in the second line. 

\subsection{Overview of the Asymptotic Results}

I will now introduce  the estimator of the identified set $\Theta_I$, leaving the formal definition to  (\ref{eq:theta0}). Suppose there exists a function $g(D,\theta,\xi)$ that preserves the expectation of the moment function $m(w,\theta,\eta)$
\begin{align*}
    \E g(D,\theta,\xi) = \E m(w,\theta,\eta)
\end{align*}
and is insensitive to the biased estimation of its own nuisance parameter $\xi$ around its true value $\xi_0=\xi_0(\theta)$, where $\xi_0(\theta)$ is an identified vector-valued parameter of the distribution $P_D$ for each $\theta \in \Theta$. In many relevant cases such as Example \ref{ex:trans}, $\xi_0$ contains the original nuisance parameter $\eta$, but may contain more unknown parameters of the distribution $P_D$.

To estimate the identified set $\Theta_I$ I represent it as the minimizer of the criterion function $Q(\theta,\xi_0)$
\begin{align}
\label{eq:loss}
    \Theta_I:=  \arg \min_{\theta \in \Theta} Q(\theta,\xi_0),
\end{align}
where $Q(\theta,\xi_0)$ is 
\begin{align*}
    Q(\theta,\xi_0) &= \| \E g(D,\theta,\xi_0) \|_{+}^2,
\end{align*}
and its sample analog is 
\begin{align}
\label{eq:qn:chap2}
    Q_N(\theta,\widehat{\xi}):= \| \frac{1}{N} \sum_{i=1}^N g(D_i, \theta, \widehat{\xi}_i)\|_{+}^2.
\end{align}

My goal is to  use different samples in the first and the second stages in order to avoid overfitting. Yet, simple sample splitting has a drawback that only one half of the sample is 
used for second-stage estimation, which can lead to loss of efficiency in small samples. In order to use the whole sample for the second stage yet keep the sample splitting idea,
I use cross-fitting procedure described below.
\begin{definition}[Cross-fitting]
 \label{sampling} 	
    \begin{enumerate}
	\item  For a random sample of size $N$, denote a $K$-fold random partition of the sample indices $[N]=\{1,2,...,N\}$ by $(J_k)_{k=1}^\mathcal{K} $, where $\mathcal{K}$ is  the number of partitions and the sample size of each fold is $n = N/\mathcal{K}$. Also for each $k \in [\mathcal{K}] = \{1,2,...,K\}$ define $J_k^c = \{1,2,...,N\} \setminus J_k$.
	\item For each $k \in [\mathcal{K}]$, construct 	an estimator $ \widehat{\xi}( V_{i \in J_k^c})$   of the nuisance parameter value $\xi_0$ using only the data from $J_k^c$.  For any observation $i \in J_k$, define an estimated signal $\widehat{\xi}_i :=  \widehat{\xi}( V_{i \in J_k^c})$.
		\end{enumerate}
\end{definition}

\begin{definition}[Definition of the Set Estimator]
\label{def:setestim}
Let $\widehat{\xi}= \widehat{\xi}(\theta)$ be the first-stage estimator of the nuisance parameter constructed in Definition \ref{sampling}. Let the criterion function $Q_N(\theta,\xi)$ be as in (\ref{eq:qn:chap2}) and $\widehat{c}$ be a positive number. The estimator $\widehat{\Theta}_I$  of the identified set $\Theta_I$ is chosen
as a contour set of the sample criterion function $Q_N(\theta,\xi)$ of level $c$:
\begin{align*}
     \widehat{\Theta}_I:= {\cal C}_N(\widehat{c},\widehat{\xi}):= \{ \theta \in \Theta, \quad N Q_N(\theta,\widehat{\xi}) \leq \widehat{c} \}.
 \end{align*}
\end{definition}
The  contour level $\widehat{c}$, possibly data dependent, is chosen such that $ \widehat{\Theta}_I$ contains the true set $\Theta_I$ with probability approaching one
 \begin{align}
 \label{eq:minc}
     \widehat{c} \geq \sup_{\theta \in \Theta_I} N Q_N(\theta,\widehat{\xi}) \text{ w.p. $\rightarrow$ 1}.
 \end{align}

 I establish convergence and inference properties of the contour set estimator $\widehat{\Theta}_I$.  The first property is formulated in terms of the convergence rate of $\widehat{\Theta}_I$ to $\Theta_I$ is based on the notion of Hausdorff distance $d_H(\widehat{\Theta}_I,\Theta_I)$:
 \begin{align*}
     d_H(\widehat{\Theta}_I,\Theta_I) = \{ \sup_{\theta \in \widehat{\Theta}_I} d_H(\theta,\Theta_I), \sup_{\theta \in \Theta_I} d_H(\widehat{\Theta}_I,\theta)  \}.
 \end{align*}
The set $\widehat{\Theta}_I$ is said to converge to $\Theta_I$ at rate $\epsilon_N$ if  the Hausdorff distance between the sets converges at rate $\epsilon_N$:  \begin{align}
 \label{eq:bestrate}
     d_H(\widehat{\Theta}_I,\Theta_I) &= O_{P} (\epsilon_N).
 \end{align}
Under mild regularity conditions it is possible to achieve the nearly efficient rate $\epsilon_N = O(\sqrt{\frac{ \log N}{N}})$.

To conduct inference, I fix a confidence level $\tau \in (0,1)$. A confidence region $C_N(c_{\tau},\widehat{\xi})$ of level $\tau$ is defined as a contour set $C_N(c_{\tau};\widehat{\xi})$
of level $c_{\tau}$ such that $C_N(c_{\tau};\widehat{\xi})$ contains $\Theta_I$ with probability at least $1-\tau$:
\begin{align}
\label{eq:conf}
    \Pr (\Theta_I \subseteq C_N(c_{\tau};\widehat{\xi}))  \rightarrow 1- \tau, N \rightarrow \infty.
\end{align}
I construct a confidence region $ C_N(\widehat{c}_{\tau};\widehat{\xi})$, where $\widehat{c}_{\tau}$ is a consistent estimator of the $\tau$-quantile, denoted $c_{\tau}$, of the inferential statistic:
\begin{align}
\label{eq:infstat}
    {\cal C}_N:= \sup_{\theta \in \Theta_I} N Q_N(\theta;\xi_0).
\end{align}
The consistent estimator of the $\tau$-quantile is chosen by the subsampling algorithm defined below.
\begin{definition}[Subsampling Algorithm]
\label{def:subsampling}
Partition the sample $(D_i)_{i=1}^N$ into $B_N= o(\sqrt{N})$ subsamples $V_j,\quad j   \in \{1,2, \dots, B_N\}$ of equal size $b:=[N/B_N]$.  Compute the sample criterion function $Q_{j,b}(\theta;\widehat{\xi})=\| \frac{1}{b} \sum_{ i \in V_j} g(D_i, \theta, \widehat{\xi}_i)\|_{+}^2.$  Choose the level $\widehat{c}$ of the order $\widehat{c} \sim \log N$. Report $\widehat{c}_{\tau}$ as the $\tau$-quantile of the sample of statistics
$$\{ \sup_{\theta \in {\cal C}_N(\widehat{c},\widehat{\xi})} b Q_{j,b}(\theta,\widehat{\xi}) , j = 1,2, \dots, B_N\}.$$
\end{definition}
Section \ref{sec2:theory} establishes the asymptotic validity of the set estimator given in Definition \ref{def:setestim} and subsampling algorithm of Definition \ref{def:subsampling}

\section{Dynamic Game of Imperfect Information}
\label{sec2:dynamic}
Consider the dynamic model of strategic interaction from \cite{BBL}. There are $K$ players, denoted by $\{1,2,\dots, K\}$. Each player $k$  makes a  decision $a_k \in \mathcal{A}_k$  from a finite set of discrete alternatives $\mathcal{A}_k$ at time periods $t \in \{0,1,\dots, \infty\}$. In each  period  $t$ the players commonly observe a vector of state variables $w_t \in \mathcal{W} \subset \mathcal{R}^{d_w}$.   Given the state variable $w_t=w$, players choose actions simultaneously. Before choosing his action, each player  $k$ observes a vector of private shocks $(\epsilon_k(j))_{j \in \mathcal{A}_k}$ corresponding to each discrete alternative $j$ in his choice set $\mathcal{A}_k$.   The transition between states follows a conditional probability distribution: $P(\cdot| a,w) $ conditional on the current state $w$ and the action profile $a=(a_1,a_2,\dots,a_K)$.

I focus on the structural parameter describing the utility of the first player. I assume that his per-period utility function is given by
\begin{align}
\label{eq:perperiod}
    \pi(a, w,\epsilon) &=\tilde{\pi} (a,w;\theta_0;\zeta_0) + \epsilon_1(a_1),
\end{align}
where $a=(a_1,a_2,\dots,a_K)$ is the profile of the players' actions, $a_1 \in  {\mathcal{A}_1}$ is the action of the first player,  $w$ is the state variable, and $\theta_0$ is the true value of the structural parameter $\theta$. The per-period utility is presented as the sum of a deterministic component $\tilde{\pi}(a,w;\theta;\zeta)$ and the private shock $\epsilon_1(a_1)$.
I allow $\tilde{\pi}(a,w;\theta;\zeta)$ to depend on identified nuisance parameter  $\zeta$ whose true value is $\zeta_0$.

I assume that each player follows a pure Markov policy. The pure Markov policy for  player one $\sigma_1(w,\epsilon_1):\mathcal{W} \bigtimes \mathcal{E} \rightarrow \mathcal{A}_1$ maps the current state $w$ and the private shock of player one, $\epsilon_1$, into the action space $\mathcal{A}_1$. If the behavior of the players is described by a Markov policy profile $\sigma = (\sigma_1,\sigma_2,\dots,\sigma_K)$, the value function is given by
\begin{align}
\label{eq:cexp}
     V(w;\theta;\sigma):= \E [\sum_{t \geq 0} \beta^t \pi (\sigma(w,\epsilon),w_t,\epsilon_t)|w].
\end{align}

A strategy profile $\sigma^{*}$ is a Markov perfect equilibrium if each player $k$ prefers its strategy $\sigma^{*}_k$ to all alternative Markov strategies as long as the others follow the equilibrium strategy $\sigma^{*}_{-k}$. That is, the value function $ V(w;\theta_0;\sigma^{*};\eta_0)$ of the first player at the strategy $\sigma^{*}$ is weakly larger than the value function of any other strategy profile $\sigma = (\sigma_1,\sigma_{-1}^{*})$:
\begin{align}
\label{eq:ineq:exante}
    V(w;\theta_0;\sigma^{*};\eta_0) \geq  V(w;\theta_0;\sigma;\eta_0) \quad \forall w \quad \forall \sigma,
\end{align}
where $\sigma = (\sigma_1,\sigma_{-1}^{*})$
consists of a feasible suboptimal alternative for player one, $\sigma_1$, and the equilibrium profile for the other players $\sigma_{-1}^{*}$.
I assume that the each observation comes from the same Markov perfect equilibrium $\sigma^{*}$, although I do not provide conditions for the existence of such equilibrium and recognize that there could be many such equilibria.
\begin{assumption}[Equilibrium Selection]
\label{ass:es}
The data are generated by a single Markov perfect equilibrium $\sigma^{*}$.
\end{assumption}

\subsection{Casting problem as a moment inequality with a first-stage nuisance parameter}
Define the choice-specific value function of player one $v(a_1,w)$ as the expected present value conditional on the current state $w$ and the choice $a_1$
\begin{align*}
    v(a_1,w):= \tilde{\pi}(a_1,w;\theta_0,\zeta_0) + \beta \E [V(w_{next};\theta_0;\sigma^{*};\eta_0)|w,(a_1,\sigma^{*}_{-1})].
\end{align*}
Then, following the optimal strategy $\sigma^{*}(w,\epsilon)$, the first player chooses $a_1$  if and only if the total utility of $a_1$ is not smaller than the total utility of any other choice $a_1' \in \mathcal{A}_1$
\begin{align*}
    v(a_1,w) + \epsilon(a_1) \geq v(a_1',w) + \epsilon(a_1').
\end{align*}
Therefore, the optimal strategy of player one is characterized as
\begin{align*}
    \sigma^{*}(w,\epsilon_1) &= \arg \max_{a_1' \in \mathcal{A}_1} \{ v(a_1',w) + \epsilon_1(a_1')  \},
\end{align*}
or, equivalently,
\begin{align*}
    \sigma^{*}(a_1,w) &= \arg \max_{a_1' \in \mathcal{A}_1} \{ v(a_1',w) - v(1,w)  + \epsilon_1(a_1')  \}.
\end{align*}
As discussed in \cite{HotzMiller}, under Assumption \ref{ass:logit} the vector of differences of the choice-specific value functions can be expressed as
\begin{align}
\label{eq:HotzMiller}
    v(a_1,w) - v(1,w) &= \log \frac{\Pr (a_1|w)}{\Pr (1|w)},
\end{align}
where $\Pr (a_1|w)$ is the probability of the choice $a_1 \in \mathcal{A}_1$ conditional on the state variable $w$. Finally, as a suboptimal Markov policy $\sigma$, I consider a cutoff-type strategy
\begin{align}
\label{eq:subopt}
    \sigma_1(w,\epsilon_1) &= \arg \max_{a_1 \in \mathcal{A}} \{ v(a_1,w) - v(1,w) +\text{dev}(a_1,w)+ \epsilon_1(a_1)\},
\end{align}
where I add a deviation function $\text{dev}(a_1,w)$ for each element $v(a_1,w) - v(1,w)$. As a normalization condition, I set $\text{dev}(1,w)=0 \quad \forall w$. Therefore, the inequality (\ref{eq:ineq:exante}) depends on the nuisance parameter
\begin{align}
\label{eq:eta}
    \eta:= (\zeta, \{\gamma_{jk} (w), j \in \mathcal{A}_k\}_{k=1}^K, \Pr (w_{next}|w,a)\}
\end{align}
that consists of any original nuisance parameter $\zeta$ that may appear in the utility function (\ref{eq:perperiod}),  the conditional choice probabilities  of all players
$$\gamma_{jk} (w):= \Pr(a_k=j|w), \quad j \in \mathcal{A}_k, k \in \{1,2,\dots,K\},$$
and the conditional distribution $\Pr (w_{next}|w,a)$.

I construct the identified set $\Theta_I$ using a subset of inequalities implied by the equilibrium definition  (i.e, (\ref{eq:ineq:exante})). Let $q(w)$ be an $L$-vector of non-negative weighting functions.
Let $\text{dev}(a_1,w): \mathcal{A}_1 \bigtimes \mathcal{W} \rightarrow \mathcal{R}^d$ be an $ L$-vector of deviation functions  whose  coordinate $l \in \{1,2,\dots,L\}$ corresponds to a deviation strategy $\text{dev}_l(a_1,w)$.
Define the moment vector-valued function as
\begin{align}
\label{eq:ineq}
    m(w,\theta,\eta)&= q(w) \cdot (V(w;\theta;\sigma;\eta) - V(w;\theta;\sigma^{*};\eta)),
\end{align}
whose  component  $m_l(D,\theta,\eta)$ is equal to  the weighted difference $V(w;\theta;\sigma_l;\eta) - V(w;\theta;\sigma^{*};\eta)$:
 $$ m_l(w,\theta,\eta) :=q_l(w) (V(w;\theta;\sigma_l;\eta) - V(w;\theta;\sigma^{*};\eta)), \quad l \in \{1,2,\dots,L\}$$
evaluated at a strategy profile $\sigma_l = (\sigma_{1,l}, \sigma^{*}_{-1})$. The suboptimal strategy of the first player  $\sigma_{1,l}$ is given in (\ref{eq:subopt}) with the deviation function $\text{dev}_l(a_1,w)$. The identified set $\Theta_I$ is defined as the collection of parameter values $\theta$ that obey inequality restrictions in expectation
\begin{align}
\label{eq:theta0}
    \Theta_I := \{ \theta \in \Theta: \quad \E m(w,\theta, \eta_0) \leq 0\},
\end{align}
where the expectation is taken with respect to the unconditional distribution of the state $w$. 
\subsection{Simulation estimator of the value function}

%  I explain why this estimator is biased in the case the dimension of the state variable is large and describe how to correct the bias.

The value function $V(w;\theta;\sigma;\eta)$ appearing in the moment function (\ref{eq:ineq}) is not available in closed form. However, it can be approximated by Monte Carlo simulation as in \cite{BBL}. In the first stage, one constructs an estimate $\widehat{\eta}$ of the nuisance parameter $\eta$ that is defined in (\ref{eq:eta}). In the second stage, one simulates  the sequence of states and shocks from the estimated first stage parameter $\widehat{\eta}$ and averages the  realization of the value function over  multiple simulation draws.  A single simulation draw is given by Algorithm \ref{alg:forward}.
\begin{algorithm}[H]

Input: initial state $w$;  parameter value $\theta$; estimated first-stage  parameter $\widehat{\eta}$; strategy profile $\sigma(w,\epsilon)$ that is a known function of the first-stage  parameter $\widehat{\eta}$. Initialize $V(w_0;\theta;\sigma;\eta)=0$.
\begin{algorithmic}[1]
\STATE Draw a shock vector $\epsilon$ from the type $1$ extreme value distribution and compute the action profile $a= \sigma(w,\epsilon)$.
\STATE Draw the state variable $w_t$ from the conditional distribution $P(w_t|w_{t-1},a_{t-1})$.
\STATE Draw a vector of shocks $\epsilon_t$  from the type $1$ extreme value distribution. Compute the action $a_t =\sigma(w_t,\epsilon_t)$.
\STATE Add the time $t$ discounted utility: $$V(w_0;\theta;\sigma;\widehat{\eta})=V(w_0;\theta;\sigma;\widehat{\eta})+ \beta^t (\tilde{\pi}(a_t,w_t;\theta;\widehat{\zeta}) +\epsilon_1(a_{1,t})).$$
\end{algorithmic}
Return $\widehat{V}(w_0;\theta;\sigma;\widehat{\eta})$.
\caption{Simulation Estimator of the Expected Value Function}
\label{alg:forward}
\end{algorithm}
In order to evaluate the value function $V(w;\theta;\sigma;\eta)$ at different parameter values $\theta_1 \in \Theta$ and $\theta_2 \in \Theta$, I use the same simulation draws. As long as the number of simulation draws is  large enough,  the simulation error does not affect the asymptotic properties of the estimator of the identified set $\widehat{\Theta}_I$ as discussed in  \cite{Pakes:1989}. 

% The proposed  plug-in estimator of the value function is due to \cite{BBL}.

\subsection{Bias Correction Term for the Expected Value Function}
When the nuisance parameter $\eta$ is high-dimensional and estimated by machine learning, the plug-in estimator of value function is biased. To make the moment function (\ref{eq:ineq}) insensitive to first-stage bias, I derive bias correction term. As discussed in \cite{Newey1994}, the bias correction term for a vector-valued nuisance parameter is equal to the sum of individual terms of individual components. Moreover, the bias correction term for the conditional  probability of choice $j$ has the product structure 
\begin{align*}
    \alpha^{CCP}_{j}(D;\theta;\gamma_j) &= \Gamma(w;\theta) (1_{\{a_1=j\}} - \gamma_j(w)).
\end{align*}
Furthermore, according to \cite{Newey1994},  the function $ \Gamma(w;\theta)$ is implicitly defined by the orthogonality condition explained below.

Let $g(D,\theta,\xi)$ be a moment function. Define the Gateaux derivative map $D_r: \Xi \bigtimes \Theta  \rightarrow \mathcal{R}^L$ as
 \begin{align*}
 	\partial_{r} \bigg\{ \E \bigg[ g(D,\theta,r(\xi - \xi_0) + \xi_0) \bigg], \quad \xi \in \Xi \bigg\},
 \end{align*}
 for all $r \in [0,1)$, which I assume exists. I also denote  the pathwise derivative of the expected moment function at the true value $\xi_0$
 \begin{align*}
  \partial_{\xi} \E g(D,\theta,\xi_0)[\xi - \xi_0] &=  \partial_0 \E g(D,\theta,r(\xi - \xi_0) + \xi_0).
\end{align*}

\begin{definition}[Neyman orthogonality of moment function]
\label{def:orthog}
The moment function $g(D;\theta;\xi)$ obeys the orthogonality condition at $\xi_0$ with respect to the nuisance realization set $\Xi_N \subset \Xi$ if the pathwise derivative $D_r[\xi - \xi_0]$ exists for all $r \in [0,1)$ and vanishes at $r=0$ for each $\theta \in \Theta$
\begin{align}
\label{eq:orthog:chap2}
	\partial_{\xi} \E g(D,\theta,\xi_0)[\xi - \xi_0]  = 0, \quad \xi \in \Xi, \theta \in \Theta.
\end{align}
\end{definition}

The definition of Neyman orthogonality requires that the moment function be insensitive to the biased estimation of $\xi$. This condition is the generalization of the orthogonality condition for point-identified models in \cite{doubleml2016} that is required to hold only at the true value $\theta_0$ of identified parameter $\theta$.  In contrast to the point-identified case, I require the equality in  (\ref{eq:orthog:chap2}) to hold at each point $\theta$ of the parameter space $\Theta$.   In many relevant cases (e.g., if  the moment function $ g(D,\theta,\xi)$ is linear in $\theta$), the orthogonality condition (\ref{eq:orthog:chap2}) on the set $\Theta$   follows from the orthogonality (\ref{eq:orthog:chap2}) on a finite subset of $\Theta$. 
Rewriting the orthogonality condition (\ref{eq:orthog:chap2}) for the bias correction term for the value function gives
\begin{align*}
     \partial_{0} \E  [V  (w;\theta; \sigma^{*};r(\gamma - \gamma_0) + \gamma_0)  - \Gamma(w;\theta) [\gamma(w) - \gamma_0(w)] =0.
\end{align*}

I find the function $\Gamma(w;\theta)$ from the recursive definition of the value function (\ref{eq:recursive}). In the case of Example \ref{ex:trans} the recursive definition can be rewritten in the unconditional form

\begin{align}
\label{eq:recursive:trans:uncond}
    \E \bigg[  V(w;\theta;\sigma^{*};\gamma_0) &- \big( - R (1-\bm{\gamma}_0(w)) + (-\mu s) \bm{\gamma}_0(w) +  PS_{\sigma^{*}} (\bm{\gamma}_0)\nonumber \\
    &+ \beta \E  [ V(w_{next};\theta;\sigma^{*};\gamma_0)|w,a=0] (1-\bm{\gamma}_0(w)) \nonumber  \\
    &+ \beta \E  [ V(w_{next};\theta;\sigma^{*};\gamma_0)|w,a=1] \bm{\gamma}_0(w) \big)\bigg] =0,
\end{align}
where $\gamma_0(w)$ is the true value of the conditional choice probability and $\theta = (R,\mu)$ consists of the replacement and maintenance costs. The unknown function $\gamma(w)$ appears in  Equation (\ref{eq:recursive:trans:uncond}) both {\bf outside} the value function $V(w;\theta;\sigma^{*};\gamma) $ and {\it inside} of this function. I consider a local deviation of the choice probability $\gamma(\cdot)$ from its  value $\gamma_0(\cdot)$ for each value of the parameters $R$ and $\mu$.
Applying chain rule to the Equation (\ref{eq:recursive:trans:uncond}) yields a pathwise derivative
\begin{align*}
    \partial_{0} \E V (w;\theta;  \sigma^{*};r(\gamma - \gamma_0) + \gamma_0) &= \frac{1}{1-\beta} \E \big(
    - \mu s + R + \beta \E  [ V(w_{next};\theta;\sigma^{*};\gamma_0)|w,a=1]\\
    &- \E  [ V(w_{next};\theta;\sigma^{*};\gamma_0)|w,a=0] +  \frac{d PS_{\sigma^{*}}(\gamma_0)}{d \gamma}
    \big)\\
    &=  \frac{1}{1-\beta}  \E \big( - \mu s + R + \beta \E  [ V(w_{next};\theta;\sigma^{*};\gamma_0)|w,a=1]\\
    &- \beta \E  [ V(w_{next};\theta;\sigma^{*};\gamma_0)|w,a=0]  - \frac{2}{\gamma(w)} - \log \frac{\gamma(w)}{1-\gamma(w)} \big),
\end{align*}
where the last equality follows from  the logistic distribution of the private shocks (Assumption \ref{ass:logit}).

In what follows I derive the bias correction terms for the other  nuisance parameters that are present in the dynamic discrete choice in \cite{BBL}. Let $q(w)$ be a weighting function. Define
\begin{align*}
    \lambda(w'):= \E[q(w)|w_{next}=w']
\end{align*}
as the expectation of the  weighting function $q(w)$ evaluated at the current state $w$ conditional  on the future state $w_{next}$. Denote   the  expectation of the choice $j$ made by player one conditional on the state $w$ by $\gamma_j(w):= \E [1_{\{a_1=j \}}|w]$. Let
$$\gamma(w):= (\gamma_2(w),\gamma_3(w),\dots, \gamma_{A_1}(w))$$ be the $A_1-1$ vector of these probabilities. Define  the conditional expectation of the current private shock evaluated at the equilibrium strategy $\sigma_1^{*} = \sigma_1^{*}(w,\epsilon,\gamma)$  as
\begin{align}
\label{psgamma}
    PS_{\sigma_1^{*}}(\gamma):= \E [\epsilon_1(\sigma_1^{*}(w,\epsilon_1,\gamma))|w].
\end{align}
Lemma \ref{lem:mylove} gives the bias correction term,
$ \alpha^{CCP}_{j}(D;\theta;\gamma_j)$, that makes the value function $V(w;\theta;\sigma^{*};\eta)$ insensitive to the biased estimation of $\gamma_j$.

\begin{lemma}[Bias correction term for own conditional choice probability]
\label{lem:mylove}
Suppose the state variable $w$ has a stationary distribution and Assumption \ref{ass:logit} holds. Suppose the conditional probability of each choice $\gamma_j(w), \quad j \in \{2,3,\dots,J\}$ is bounded away from zero and one. Then the bias correction term $\alpha^{CCP}_{j}(D;\theta;\gamma_j)$ is 

%\begin{align}
%\label{eq:alphajccp}
%    \alpha^{CCP}_{j}(D;\theta;\gamma_j) &= \frac{q(w)}{q(w) - \beta \lambda(w)} \big( \E [\tilde{\pi} ((j,\sigma^{*}_{-1}(w,\epsilon_{-1}));w;\theta) - \tilde{\pi}(1,\sigma^{*}_{-1}(w,\epsilon_{-1});w;\theta)] \\
%    &+ \beta \E[ V_{\sigma^{*}}(w_{next};\theta;\eta_0) | w,(j,\sigma^{*}_{-1}(w,\epsilon_{-1}))] \nonumber \\
%    &- \beta \E[ V_{\sigma^{*}}(w_{next};\theta;\eta_0) | w,(1,\sigma^{*}_{-1}(w,\epsilon_{-1}))] \nonumber \\
%    &+ \partial_{\gamma_j}PS_{\sigma_1^{*}}(\gamma_0) \big) (1_{a_1=j} - \gamma_j(w))\nonumber,
%\end{align}

\begin{align}
    \alpha^{CCP}_{j}(D;\theta;\gamma_j) &= \frac{q(w)}{q(w) - \beta \lambda(w)} \big( \E_{\epsilon_{-1}} [\tilde{\pi} ((j,\sigma^{*}_{-1}(w,\epsilon_{-1}));w;\theta) - \tilde{\pi}(1,\sigma^{*}_{-1}(w,\epsilon_{-1}));w;\theta)] \nonumber \\
    &+ \beta \E[ V(w_{next};\theta;\sigma^{*};\eta_0) | w,(j,\sigma^{*}_{-1}(w,\epsilon_{-1}))] \nonumber \\
    &- \beta \E[ V(w_{next};\theta;\sigma^{*};\eta_0) | w,(1,\sigma^{*}_{-1}(w,\epsilon_{-1}))] \nonumber \\
    &+ \partial_{\gamma_j}PS_{\sigma_1^{*}}(\gamma_0) \big) (1_{a_1=j} - \gamma_j(w)) \label{eq:alphajccp}
\end{align}
where $\partial_{\gamma_j}PS_{\sigma_1^{*}}(\gamma)$ is the partial derivative of (\ref{psgamma}) with respect to $\gamma_j$.
\end{lemma}

%That the bias correction term

\begin{lemma}[Bias Correction Term for the Law of Motion of the State Variable]
\label{lem:trans}
Suppose the high-dimensional state variable $w$  has a  stationary distribution. Let $a=(a_1,a_2,\dots,a_K)$ be a given profile of actions. Then the bias correction term for the conditional quantile function $Q(u,w,a)$ of level $u$ is 
\begin{align}
\label{eq:trans}
    \alpha^{TRANS}_{\sigma}(D;\theta;\eta) &= \frac{\beta l(w)}{l(w) - \beta \lambda(w)} \E[ \nabla_{w_{next}} V (w_{next};\theta;\sigma;\eta_0)|w, a] \prod_{k=1}^K \Pr (\sigma_k(w,\epsilon_k)=a_k|w) \\
     &\frac{1_{w_{next} \leq Q(u,w,j)  - u}  }{f(w_{next}|w,a=\sigma(w,\epsilon))} \nonumber. 
\end{align}
\end{lemma}

\begin{lemma}[Bias Correction Term for the opponent's conditional choice probability]
\label{lem:opponent}
Suppose the state variable $w$ has a stationary distribution. Let the number of players be $2$ (i.e, $K=2$).  Then the bias correction term for the conditional probability of the choice $j_2$ made by my opponent is
\begin{align}
\nonumber
    \alpha_{j}^{CCP,op}(D;\theta;\gamma_{j_2}) &= \frac{q(w)}{q(w) - \beta \lambda(w)} \big(  \sum_{j_1=1}^{A_1} (\tilde{\pi} ((j_1,j_2),w;\theta) -\tilde{\pi} ((j_1,1),w;\theta )) \Pr (\sigma_1(w,\epsilon_1) = j_1|w) \\
    &+ \beta  \big[  \sum_{j_1=1}^{A_1} \E [V(w_{next};\theta;\sigma^{*};\eta_0) | w, (j_1,j_2)] - \sum_{j_1=1}^{A_1} \E [V(w_{next};\theta;\sigma^{*};\eta_0) | w, (j_1,1)] \big] \cdot \nonumber \\
    &\cdot \Pr (\sigma_1(w,\epsilon_1) = j_1|w)) (1_{a_2 = j_2} - \Pr(a_2=j_2) ). \label{eq:opponent}
\end{align}
\end{lemma}

\subsection{Using Linearity to Reduce Computation}

Both the value function and the bias correction terms presented above are not available in closed form and must  be simulated. When the value function is a linear  function of $\theta$, the simulation can be simplified. Suppose there exist  basis functions $\Psi_1(w;\sigma;\eta), \Psi_2(w;\sigma;\eta) $ such that
\begin{align}
\label{eq:basis}
V(w;\theta;\sigma;\eta) = \theta \Psi_1(w;\sigma;\eta) + \Psi_2(w;\sigma;\eta) 
\end{align}
is an affine function of $\theta$. Then one can simulate the basis functions $\Psi_1(w;\sigma;\eta),  \Psi_2(w;\sigma;\eta) $ instead of simulating the value function for each $\theta \in \Theta$. Lemma \ref{lem:linearity} provides the sufficient conditions for the linearity of the value function and the individual bias correction terms provided in Lemmas \ref{lem:mylove}, \ref{lem:trans}, \ref{lem:opponent}.

\begin{lemma}[Sufficient Conditions for Linearity]
\label{lem:linearity}
The following conditions hold. (1)  The per-period utility function given in (\ref{eq:perperiod}) is a linear function of $\theta$.  (2) The distribution of the private shock for each player is known.   Then there exists a vector of basis functions $\Psi(w;\sigma;\eta)$ such that (\ref{eq:basis}) holds. Furthermore, each individual bias correction term given in (\ref{eq:alphajccp}), (\ref{eq:trans}), and (\ref{eq:opponent}) is also a linear function $\theta$.
\end{lemma}

\section{Asymptotic Theory}
\label{sec2:theory}
Suppose I have a collection of inequality restrictions on an economic model coming from the data structure and/or the assumptions about the  data generating process. These restrictions are embodied into a moment function $g(D,\theta,\xi): \mathcal{D} \bigtimes \Theta \bigtimes \Xi \rightarrow \mathcal{R}^L$,
where $\theta \in \Theta \subset \mathcal{R}^d$ is the target parameter. In addition to the target parameter $\theta$, the moment function $g(D,\theta,\xi)$ depends on a nuisance parameter $\xi=\xi(\theta)$ whose true value  $\xi_0=\xi_0(\theta)$  is an identified parameter of the data distribution  $P_D$  and belongs to a convex subset of a normed vector space $ \Xi$. The object of interest is an identified set $\Theta_I$ defined as a collection of parameter values $\theta$ that satisfy the inequality restrictions
\begin{align}
\label{eq:theta0}
    \Theta_I :=  \{ \theta \in \Theta: \quad \E g(D,\theta, \xi_0) \leq 0\}
\end{align}
at the true value $\xi_0$ of the nuisance parameter.

The identified set $\Theta_I$ is characterized as the minimizer of the criterion function
\begin{align*}
\Theta_I :=  \arg \min_{\theta \in \Theta}  Q(\theta, \xi_0) = \arg \min_{\theta \in \Theta} \| \E g(D,\theta, \xi_0) \|^2.
\end{align*}
I assume that the following partial identification condition holds. There exist positive constant $C>0$ and $\delta>0$ such that 
\begin{align}
\label{ass:partid}
	\| \E g(D,\theta, \xi_0 )\|_{+} \geq C d_H(\theta, \Theta_I) \wedge \delta.
\end{align}
This condition states that once $\theta$ is bounded away from $\Theta$, the moment  $\E g(D,\theta, \xi_0 )$ is bounded away from $\Theta_I$ by a number that is proportional to the Hausdorff  distance $d_H(\theta,\Theta_I)$  from $\theta$ to the identified set $\Theta_I$. This condition ensures that the true moment function $ g(D,\theta, \xi_0) $ distinguishes the boundary of the identified set.

 \paragraph{Impact of the First-Stage Estimation.}
In the next condition I introduce a sequence of neighborhoods $\Xi_N^{\theta} \subset \Xi^{\theta}$  of $\xi_0(\theta)$ that contain the estimate  $\widehat{\xi}(\theta)$ of $\xi_0(\theta)$ w.p. approaching one. As the sample size $N$ increases, the neighborhoods shrink. The quality of the estimation of the first-stage parameter is defined as the speed of  shrinkage of the neighborhood $\Xi_N(\theta)$ around $\xi_0(\theta)$. I refer to it as the first-stage rate $g_N$. Finally, I assume that  the second-order derivative of the functional $\E g(D,\theta, \xi)$ is well-behaved. Combined with the orthogonality condition and 
the upper bound on the first-stage rate $g_N$, this assumption ensures that  I can ignore the impact of the estimation error $\widehat{\xi}(\theta)-\xi_0(\theta)$ on the second and the higher-order derivatives of the moment function $\E g(D,\theta, \xi)$ with respect to $\xi$ at $\xi_0$.

\begin{condition}[Orthogonality]
\label{ass:smallbiasvalue}
There exists a sequence $\Xi_N^{\theta}$ of subsets of $\Xi^{\theta}$: $\Xi_N^{\theta} \subset \Xi^{\theta}$ such that the following conditions hold. (1) The true value $\xi_0(\theta)$ belongs to $\Xi_N^{\theta}$ for all $N \geq 1$. (2) There exists a sequence of numbers $\phi_N=o(1)$ such that with probability at least $1-\phi_N$, $\widehat{\xi}(\cdot)$ belongs to $\Xi:=\bigtimes \Xi^{\theta}$ uniformly over $\theta \in \Theta$. (3) The set $\Xi_N^{\theta}$ shrinks around $\xi(\theta)$ at the following statistical rate  $g_N = o(N^{-1/4})$ uniformly over $\theta$:
$$ \sup_{\theta \in \Theta} \sup_{\xi(\theta) \in \Xi^{\theta}} \| \xi(\theta) - \xi_0(\theta)\|_{P,2} \leq g_N.$$  The moment function $g(D,\theta,\xi)$ obeys the orthogonality condition at $\xi_0$.  (5) There exists a sequence $s_N = o(N^{-1/2})$ such that the  second Gateaux derivative of the functional $G(\theta,\xi(\theta))$ with respect to $\xi$ at $\xi_0$ is bounded:
$$ \sup_{\theta \in \Theta} \sup_{\xi \in \Xi_N^{\theta}} \sup_{r \in [0,1)} \| \partial_r^2 \E g(D,\theta, r(\xi - \xi_0) + \xi_0) \| \leq s_N.$$
\end{condition}

The next condition requires that the moment function $g(D;\theta;\xi(\theta))$ is sufficiently regular with respect to $\theta$ for each fixed element $\xi \in \Xi$ of the nuisance realization set $\Xi$. I consider a class $$\mathcal{F}_{\xi}:= \{ g(\cdot,\theta, \xi(\theta)), \theta \in \Theta \}$$
and require the uniform covering entropy of this class to be bounded.  

\begin{condition}[Regularity of Moment Function]
\label{ass:concentration:chap2}
The following conditions hold. (1) There exists a measurable envelope function $F_{\xi}=F_{\xi}(D)$ that bounds all elements in the function class almost surely
$$ \sup_{\theta \in \Theta} | g_l(D,\theta, \xi)| \leq F_{\xi} (D) \text{ a.s. }, \quad l \in \{1,2,\dots,L\}.$$ Moreover, the envelope $F_{\xi}$ has a finite $c$-norm for some $c>2$ $\| F_{\xi} \|_{P,c} := \left(\int_{D \in \mathcal{D}} |F_{\xi}(D)|^{c} \right)^{1/c} \leq c_1$.
(2) There exist finite constants $a$ and $v$ such that the uniform covering entropy of the class $\mathcal{F}_{\xi}$ is bounded
\begin{align}
\label{eq:uniform:chap2}
\sup_{\tilde{Q} } \log N(\epsilon \| \mathcal{F}_{\xi} \|_{\tilde{Q} ,2}, F_{\xi}, \| \cdot \|_{\tilde{Q} ,2} \leq v \log (a/\epsilon), \text{ for all } 0 < \epsilon \leq 1.
\end{align}
(3) There exists a sequence $r_N'$ obeying $r_N' \log (1/r_N') = o(1)$ that is an upper bound for the following quantity
$$ \sup_{\theta \in \Theta} \sup_{\xi \in \Xi^{\theta}} \left ( \E \| g(D, \theta, \xi(\theta)) - g(D, \theta, \xi_0(\theta)) \|^2 \right)^{1/2} \leq r_N'.$$
\end{condition}

Conditions \ref{ass:concentration:chap2}(1)-(2)  are the  generalization of the regularity assumption in the point-identified moment problem of \cite{doubleml2016}. Because
 $\xi (\theta)$ depends on $\theta$, conditions \ref{ass:concentration:chap2}(1)-(2) are  non-standard and require verification in applications. However, this requirement is mild when the nuisance parameter  $\xi (\theta)$  is a linear function of $\theta$.  Suppose that the true value of the nuisance parameter $\xi_0(\theta)$ is a linear function of $\theta$
\begin{align}
\label{eq:affine:chap2}
\xi_0(\theta):= \xi_0^{a}(D)'\theta + \xi_0^{b}(D),
\end{align}
where  $\xi_0^{a}(D)$ and  $\xi_0^{b}(D)$ are the identified parameters of the distribution $P_D$. Then conditions \ref{ass:concentration:chap2}(1)-(2) can be reformulated in terms of the nuisance parameters $\{  \xi_0^{a}, \xi_0^{b}\}$ that no longer depend on $\theta$. When (\ref{eq:affine:chap2}) holds, Conditions \ref{ass:concentration:chap2} (1)-(2) are satisfied for many practical cases. In particular, the functions $\xi_0^{a}(D), \xi_0^{b}(D)$ can be estimated by $\ell_1$-regularized methods, random forests, and deep neural nets under plausible assumptions about their structure.

\paragraph{Donsker Property.} Let $\Theta'$ be an open neighborhood of $\Theta$. I require the moment function $g(D_i,\theta,\xi_0)$ to have a Donsker property defined as follows. In the metric space $L^{\infty}(\Theta')$,
\begin{align}
\label{eq:pdonsker}
 \G_N g(D_i,\theta, \xi_0 ):=\sqrt{N} (\E_N g(D_i,\theta, \xi_0) - \E g(D_i,\theta, \xi_0)) \Rightarrow \Delta(\theta),
\end{align}
where $\Delta(\theta)$ is a mean zero Gaussian process on $\Theta$ with a.s. continuous paths and $\text{Var}(\Delta(\theta))>0$ for each $\theta \in \Theta'$. In addition, the probability space $(\Omega, \mathcal{F}, \mathcal{P})$ is rich enough to support the representation (\ref{eq:pdonsker}). 

\begin{theorem}[Estimation and Inference for Semiparametric Functional Inequalities]
\label{thm:main:ineq:chap2}
Suppose Conditions (\ref{ass:partid}), \ref{ass:smallbiasvalue},  \ref{ass:concentration:chap2}, and (\ref{eq:pdonsker}) hold. Let $\widehat{\Theta}_I$ be a contour set  estimator of Definition \ref{def:setestim}. Let $\widehat{c}$ be such that 
\begin{align}
\label{eq:minc:inf}
	\widehat{c} \geq \sup_{ \theta \in \Theta_I} N Q_N(\theta, \xi_0) \quad \text{ w.p. } \rightarrow 1
\end{align}
  holds.  Then the Hausdorff distance  $d_H(\Theta_I, \widehat{\Theta}_I)$  between the estimated set $\widehat{\Theta}_I$ and $\Theta_I$ converges  at rate $O_{P}(\sqrt{(1\vee \widehat{c})/N})$: $d_H(\widehat{\Theta}_I, \Theta_I) = O_{P}(\sqrt{(1\vee \widehat{c})/N})$.
\end{theorem} 

Theorem \ref{thm:main:ineq:chap2} is my first main result. It establishes the sufficient conditions on the moment function to deliver the rate of   convergence  of $\widehat{\Theta}_I$ to $\Theta_I$.  It suggests that the contour level $\widehat{c}$ as small as possible subject to the constraint (\ref{eq:minc:inf}).  Setting $\widehat{c} = O_{P} (1)$ subject to (\ref{eq:minc:inf}) delivers the optimal rate, but this choice is infeasible. Setting $\widehat{c} \sim \log N$ delivers a nearly efficient rate.

In many cases it is possible to establish convergence without the requirement (\ref{eq:minc:inf}). This is possible because the criterion function is degenerate  (\ref{eq:qn:chap2}).
\begin{definition}[Degeneracy]
\label{def:degeneracy}
The following conditions hold. (1)  There exists a sequence of subsets  $\Theta_N$ of $\Theta$, which cannot depend on $\xi$, such that the criterion function $Q_N(\theta, \xi)$ vanishes on these subsets w.p. approaching one. That is, $\forall p > 0$ there exists $N_p$ such that for all $N \geq N_P$
$ \inf_{\xi \in \Xi_N} \Pr (Q_N(\theta, \xi) - \inf_{\theta \in \Theta} Q_N(\theta,\xi) = 0 \quad \forall \theta \in \Theta) \geq 1-p $. (2)  These sets can approximate the identified set $\Theta_I$ in the Hausdorff distance sufficiently well: $d_H(\Theta_N,\Theta_I) \leq \epsilon_N$. (3) The sequence $\epsilon_N = O_{P} (N^{-1/2})$.

\end{definition}

\begin{lemma}[Sufficient Conditions for Degeneracy]
\label{lem:degeneracy}
Suppose Conditions (\ref{ass:partid}), \ref{ass:smallbiasvalue},  \ref{ass:concentration:chap2}, and (\ref{eq:pdonsker})   hold. In addition, there exist positive constants $C,M,\delta$ such that
\begin{align}
\label{eq:deg}
    \max_{l} \E g_l(D,\theta, \xi_0 )  \leq - C (\epsilon \wedge \delta) \quad \text{for all } \theta \in \Theta_I^{-\epsilon}, \\
    d_H(\Theta_I^{-\epsilon},\Theta_I) \leq M \epsilon \text{ for all } \epsilon \in [0, \delta], \quad l \in \{1,2,\dots, L\} \nonumber
\end{align} 
Then the criterion function $ Q_N(\theta, \xi)$  obeys the degeneracy condition in the sense of Definition \ref{def:degeneracy}. Suppose the contour level $\widehat{c}$ obeys 
\begin{align}
\label{eq:minc:deg}
\widehat{c} \geq \min_{\theta \in \Theta} Q_N(\theta, \widehat{\xi}) \vee \frac{\log N}{\sqrt{N}} \text{ w.p. } \rightarrow 1.
\end{align}
Then the Hausdorff distance  $d_H(\Theta_I, \widehat{\Theta}_I)$  converges  at rate $O_{P}({N}^{-1/2})$, where $\widehat{\Theta}_I$ is a contour level set as in Definition \ref{def:setestim}.
\end{lemma}

Lemma \ref{lem:degeneracy} is my second main result. It provides the sufficient conditions under which the identified set $\Theta_I$ can be estimated at the fastest possible rate $\epsilon_N = O_{P} (N^{-1/2})$. It requires that the sample criterion function $Q_N(\theta,\xi_0)$ be flat on the (possibly) data-dependent sets $\Theta_N$ that approximate the identified set $\Theta_I$ sufficiently well. If this requirements holds,  the contour level $\widehat{c}_0:= \arg \min_{\theta \in \Theta_0} Q_N(\theta,\xi_0)$ delivers the optimal rate for the contour level set ${\cal C}_N(\widehat{c}_0, \xi_0)$ based on the true value of the nuisance parameter  (see, e.g. \cite{CHT}). We show that a modified choice $\widehat{c}$ given in (\ref{eq:minc:deg}) delivers the optimal rate in the presence of the nuisance parameter $\widehat{\xi}$ estimated in the first stage on an auxiliary sample.

\paragraph{Subsampling.}  I wish to construct the contour set ${\cal C}(\widehat{c}, \widehat{\xi})$ that has confidence region property (\ref{eq:conf}). To do this, I must find the asymptotic distribution of the inferential statistic 
$$ {\cal C}_N:= \sup_{\theta \in \Theta_I} Q(\theta; \xi_0)$$
that can be used to estimate the $\tau$-quantile of $ {\cal C}_N$. Define the random variable
\begin{align}
\label{eq:infstat:moment}
 \mathcal{C}:= \sup_{\theta \in \Theta_I} \| I(\theta) + \Delta(\theta) \|_{+}^2,
\end{align}
 where $I(\theta)$ is an $L$-vector of functions $I_l(\theta), l \in \{ 1,\dots,L\})$. The function $I_l(\theta)$ is defined as follows:
 \begin{align*}
 I_l(\theta):= \begin{cases} - \infty, \quad \E g_l(D,\theta,\xi_0)<0, \\
 0, \quad \E g_l(D,\theta,\xi_0)=0.
 \end{cases}
 \end{align*}
 \cite{CHT} show that  ${\cal C}_N$ converges to $ \mathcal{C}$ in distribution, and that $ \mathcal{C}$ has a non-degenerate and continuous distribution function. The final requirement for the validity of the subsampling algorithm of Definition \ref{def:subsampling} is that the inferential statistic 
\begin{align*}
\sup_{\theta \in \Theta_I^{\epsilon_N}}  {\cal C}_N (\theta, \widehat{\xi}) 
\end{align*}
is well-behaved on the $\epsilon_N$-expansion of the identified set $\Theta_I$. A sufficient condition for this requirement is given below.

\begin{condition}[Sufficient Conditions for Subsampling]
\label{ass:dominance}
(1) There exists a constant $C_{\max} < \infty$ such that the moment function is bounded: $\sup_{\theta \in \Theta} | \E g(D,\theta,\xi_0) | \leq C_{\max} d_H(\theta, \Theta_I)$. (2) The rates $s_N, r_N'$ and $\epsilon_N$ obey the following bound: $\sqrt{N} (s_N + r_N' \log (1/r_N') + N^{-1/2+1/c} ) \epsilon_N = o(1)$.
\end{condition}

\begin{theorem}[Validity of Subsampling for Moment Inequalities]
\label{thm:main:subs}
Suppose Conditions (\ref{ass:partid}), \ref{ass:smallbiasvalue},  \ref{ass:concentration:chap2},  (\ref{eq:pdonsker}), \ref{ass:dominance} hold. Suppose the number of subsamples satisfies $b = o (\sqrt{N})$, $b  \rightarrow \infty$.  Let $\tau $ be the desired coverage level. Then  (1) the critical value $\widehat{c}$ of Definition \ref{def:subsampling} converges in probability to the $\tau$-quantile of $\mathcal{C}$, where  $\mathcal{C}$ given in (\ref{eq:infstat:moment}) and (2) $\Pr (\Theta_I \subseteq \mathcal{C}_N(\widehat{c}, \widehat{\xi})) = 1 - \tau$.\end{theorem}
%
% the following statements hold:
%\begin{itemize}
%    \item  Convergence Rate. The Hausdorff distance converges to zero 
%    \item Validity of Subsampling. Define the limiting  distribution of the inferential statistic as
%    \begin{align*}
%        \mathcal{C}:= \sup_{\theta \in \Theta_I} \| I(\theta) + \Delta(\theta) \|_{+}^2,
%    \end{align*}
%   Let $\tau \in (0,1)$ be a desired coverage level. Then, the contour set $C_n(\widehat{c}, \widehat{\xi})) $ of Definition \ref{def:subsampling}  is a valid confidence region for $\Theta_I$ of level $\tau$.
%    \end{itemize}

\section{Appendix}
\label{appendix:a}
\paragraph{Notation.} 

We use standard notation for numeric and stochastic dominance. For two numeric sequences denote $\{a_{n}, b_{n}\}, n \geq 1: a_{n} \lesssim  \sqrt{n}$ stands for $a_{n} = O (b_{n})$. For two sequences of random variables denote $\{a_{n}, b_{n}, n \geq 1\}: a_{n} \lesssim_{P}  \sqrt{n} $ stands for $  a_{n} = O_{P} (b_{n})$. Finally,   let $a \wedge b = \min \{ a, b\}$ and  $a \vee b = \max \{ a, b\} $. We use the standard notation for empirical processes. Denote the sample average  $\EN [f_i]:= \frac{1}{N} \sum_{i=1}^N f(D_i)$ and the centered, scaled sample average as $\GN [f_i]:= \frac{1}{\sqrt{N}} \sum_{i=1}^N (f(D_i) - \E f(W_i))$. Denote the sample average within a partition $I_k$ of size $n$ as $\En [f(D_i)] := \frac{1}{n} \sum_{i \in I_k} f(D_i)$ and $\Gn [f(D_i)] := \frac{1}{\sqrt{n}} \sum_{i \in I_k} (f(D_i) - \E [f(D_i)])$. The Hausdorff distance between sets $A$ and $B$ is defined as $d_H(A,B):= \max \{ \sup_{a \in A} d(a,B), \sup_{b \in B} d(b,A)  \}$ where $d(b,A):= \inf_{a \in \mathcal{A}} \| b - a\|$ and $d_H(A,B) = \infty$ if either of the sets is empty. The $\epsilon$-expansion of the set $\Theta_I$ is defined as $\Theta_I^{\epsilon} := \{ \theta \in \Theta: d(\theta, \Theta_I) \leq \epsilon \}$ and the $\epsilon$-contraction of $\Theta_I$ is defined as $\Theta_I^{-\epsilon}:= \{ \theta \in \Theta_I: d(\theta, \Theta \setminus \Theta_I) \geq \epsilon \}$. Let $\| x\|_{+} = \| \max (x,0)\|$ and $\| x \|_{-} = \| \max (-x,0) \|$, where the operation $\max$ is performed elementwise.

\begin{proof} [Proof of Remark \ref{rm:trans}]

\textbf{Step 1. Bias correction term for the conditional choice probability. }
I consider the value function $V(w;\theta;\sigma^{*};\eta)  =V(w;\theta;\sigma^{*};\gamma)   $ as a function of the  current state variable $w$, parameter $\theta$, and the conditional choice probability $\gamma= \gamma(w)$ holding the other value of the nuisance component $\rho$ at its true value $\rho_0$.
I wish to construct a bias correction term
$\alpha^{CCP}(D;\theta;\gamma)$ such that orthogonality condition (\ref{eq:orthog:chap2}) holds:
\begin{align*}
    \partial_{\gamma}\E [V(w;\theta;\sigma^{*};\gamma_0)  + \alpha^{CCP}(D;\theta;\gamma_0)] = 0 \quad \forall \theta  \in \Theta.
\end{align*}

Denote the expected value of the current private shock for the optimal strategy $\sigma^{*}$ as
\begin{align*}
    PS_{\sigma^{*}} (\gamma):= \E [\epsilon(\sigma^{*}(w,\epsilon,\gamma))|w],
\end{align*}
where by logistic distribution of the private shocks I get:
\begin{align}
\label{eq:eps}
    PS_{\sigma^{*}} (\gamma):= \E (\epsilon(1) - \epsilon(0)) 1_{\{\epsilon(1) - \epsilon(0) \geq -\log \frac{\gamma(w)}{1-\gamma(w)}\}} = -\log \gamma(w) - (1-\gamma(w)) \log \frac{1-\gamma(w)}{\gamma(w)}.
\end{align}
I rewrite the recursive equation (\ref{eq:recursive}) for the Example \ref{ex:trans}:
\begin{align}
    \label{eq:recursive:trans}
    V(w;\theta;\sigma^{*};\gamma)  &= - R (1-\gamma(w)) + (-\mu s) \gamma(w) +  PS_{\sigma^{*}} (\gamma)\\
    &+ \beta \E  [ V(w_{next};\theta;\sigma^{*};\gamma)|w,a=0] (1-\gamma(w)) + \beta \E  [ V(w_{next};\theta;\sigma^{*};\gamma)|w,a=1] \gamma(w) \nonumber.
\end{align}
I seek for the bias correction term $\alpha^{CCP}(D;\theta;\gamma)$  that has the following form:
\begin{align}
\label{eq:adj:gamma}
    \alpha^{CCP}(D;\theta;\gamma) &= \Gamma(w;\theta) (1_{\{a=1\}} - \gamma(w)),
\end{align}
where $\Gamma(w;\theta)$ is determined by the following property:
\begin{align}
\label{eq:adj:gamma:riez}
    \partial_0 \E V (w;\theta; \sigma^{*};r(\gamma - \gamma_0) + \gamma_0)  = \E \Gamma(w;\theta)[\gamma(w) - \gamma_0(w)].
\end{align}
I find the function $\Gamma(w;\theta)$ by the application of the implicit function theorem to the recursive equation (\ref{eq:recursive:trans}). I rewrite  (\ref{eq:recursive:trans}) in the unconditional form:
\begin{align*}
    \E \bigg[  \underbrace{V(w;\theta;\sigma^{*};\gamma) }_{S_1(V(w;\sigma^{*};\gamma))} &- \big( \underbrace{- R (1-\gamma(w)) + (-\mu s) \gamma(w) +  PS_{\sigma^{*}} (\gamma)}_{S_2(\gamma)}\\
    &+ \underbrace{\beta \E  [ V(w_{next};\theta;\sigma^{*};\gamma)|w,a=0] (1-\gamma(w))}_{S_3(V(w;\sigma^{*};\gamma),\gamma)} + \underbrace{\beta \E  [ V(w_{next};\theta;\sigma^{*};\gamma)|w,a=1] \gamma(w)}_{S_4(V(w;\sigma^{*};\gamma),\gamma)} \big)\bigg] \\
    &=0.
\end{align*}
Recognize that $\gamma(w)$ appears in each summand of the expression above: both inside the value function
\begin{align*}
    S_1(V(w;\sigma^{*};\bm{\gamma})),S_3(V(w;\sigma^{*};\bm{\gamma}),\gamma),S_4(V(w;\sigma^{*};\bm{\gamma}),\gamma)
\end{align*}
and outside of it:
 $$S_2(\bm{\gamma}),S_3(V(w;\sigma^{*};\gamma),\bm{\gamma}),S_4(V(w;\sigma^{*};\gamma),\bm{\gamma}).$$
Taking the derivative with respect to the outside presence of $\gamma$ gives:
\begin{align*}
   \mathbf{S}_{out} &= \E \big(-\mu s + R  +  \beta \E [V(w_{next};\theta;\sigma^{*};\gamma)|w,a=1] - \beta \E [V(w_{next};\theta;\sigma^{*};\gamma)|w,a=0] \\
    &+ \frac{dPS_{\sigma^{*}}(\gamma_0)}{d\gamma_0} \big) [\gamma(w) - \gamma_0(w)].
\end{align*}
Taking the derivative of the expected current shock $PS_{\sigma^{*}}(\gamma)$ defined in (\ref{eq:eps}) with respect to $\gamma$ gives:
\begin{align*}
    \frac{dPS_{\sigma^{*}}(\gamma)}{d\gamma} &= -\frac{2}{\gamma(w)} - \log \frac{\gamma(w)}{1-\gamma(w)}.
\end{align*}
Therefore, the pathwise (Gateaux) derivative with respect to the outside component is:
\begin{align*}
    \mathbf{S}_{out} &= \E \Gamma(w;\theta) [\gamma(w) - \gamma_0(w)].
\end{align*}
Taking the derivative inside the value function requires some preparation.
Recognize that:
\begin{align}
\label{eq:trans:change}
     \beta \E [S_3(V(w_{next};\sigma^{*};\gamma);\gamma_0)] &=
    \beta \E \E  [ V(w_{next};\theta;\sigma^{*};\gamma)|w,a=0] (1-\gamma_0(w)) \\
    &= \beta \E V(w_{next};\theta;\sigma^{*};\gamma) \frac{1_{\{a=1\}}}{1-\gamma_0(w)} (1-\gamma_0(w)) \nonumber \\
    &=\beta  \E  [ V(w_{next};\theta;\sigma^{*};\gamma) 1_{\{a=1\}}\nonumber  \\
    &=\beta \E V(w_{next};\theta;\sigma^{*};\gamma) \lambda_0(w_{next}),\nonumber
\end{align}
where
\begin{align*}
    \lambda_j(w_{next}):= \E [1_{\{a=j\}}|w_{next}], \quad j \in \mathcal{A}
\end{align*}
is the  expectation of the current choice $1{\{a=j\}}$ conditional on the future state $w_{next}$.
Assuming that the state variable has a stationary distribution, I get:
\begin{align*}
    \E V(w_{next};\theta;\sigma^{*};\gamma) \lambda_0(w_{next}) =  \E V(w;\theta;\sigma^{*};\gamma)  \lambda_0(w),
\end{align*}
which implies:
\begin{align}
\label{eq:s3}
      \beta \E [S_3(V(w_{next};\sigma^{*};\gamma);\gamma_0)] &=\E V(w;\theta;\sigma^{*};\gamma)  \lambda_0(w).
\end{align}
The same argument applied to $S_4$ delivers:
\begin{align}
\label{eq:s4}
      \beta \E [S_4(V(w_{next};\sigma^{*};\gamma);\gamma_0)] &=\E V(w;\theta;\sigma^{*};\gamma)  \lambda_1(w).
\end{align}
Combining the (\ref{eq:s3}) and (\ref{eq:s4}) and recognizing that $   \sum_{j \in \mathcal{A}}\lambda_j(w) = 1$ gives:
\begin{align*}
     \beta \E [S_3(V(w_{next};\sigma^{*};\gamma);\gamma_0)+ S_4(V(w_{next};\sigma^{*};\gamma);\gamma_0)] &= \beta \E V(w;\theta;\sigma^{*};\gamma) .
\end{align*}
Therefore,
\begin{align*}
    \partial_{\gamma} \E [S_3(V(w_{next};\sigma^{*};\gamma);\gamma_0)+ S_4(V(w_{next};\sigma^{*};\gamma);\gamma_0)] &= \beta \partial_{\gamma}  \E V(w;\theta;\sigma^{*};\gamma) .
\end{align*}
Plugging this result into (\ref{eq:recursive:trans}) gives:
\begin{align*}
     \mathbf{S}_{ins} := (1-\beta)\partial_{0}\E V(w;\theta;\sigma^{*};r(\gamma-\gamma_0)+\gamma_0).
\end{align*}
According to the chain rule, the sum of the derivatives with respect to the inside and the outside presence of $\gamma$ around $\gamma_0$ is equal to zero:
\begin{align*}
     \mathbf{S}_{out} +  \mathbf{S}_{ins} & =0,
\end{align*}
which implies that
\begin{align*}
    \partial_{0} \E V(w;\sigma^{*};r(\gamma-\gamma_0)+\gamma_0) &= -\frac{1}{1-\beta} \E \Gamma(w;\theta) [\gamma(w) - \gamma_0(w)].
\end{align*}
Therefore, the function $\Gamma(w;\theta)$ is equal to:
\begin{align*}
   \Gamma(w;\theta)&=-\mu s + R  +  \beta \E [V(w_{next};\theta;\sigma^{*};\gamma)|w,a=1] - \beta \E [V(w_{next};\theta;\sigma^{*};\gamma)|w,a=0]  \\
   &- \frac{2}{\gamma(w)} - \log \frac{\gamma(w)}{1-\gamma(w)}
\end{align*}
and the bias correction term for the conditional choice probability is equal to:
\begin{align}
    \label{eq:gamma:final}
    \alpha^{CCP}(D;\theta;\gamma):=  -\frac{1}{1-\beta} \Gamma(w;\theta) \left( 1_{\{a=1\}} - \gamma(w) \right).
\end{align}
\paragraph{Bias correction term for the transition function.}

I consider the value function $V(w;\theta;\sigma;\eta)$ as a function of the current state $w$, the parameter $\theta$ and the transition function $\rho(\cdot)$ given in (\ref{eq:trans}) holding the conditional probability of replacement $\gamma = \gamma_0$ fixed at its true value.
Since the estimation of the nuisance parameter $\rho$ affects the value function $V(w;\theta;\tilde{\sigma};\eta)$ evaluated for both the optimal and a suboptimal strategy (i.e, $\tilde{\sigma} \in \{ \sigma, \sigma^{*}\}$),  I focus on a Markov policy $\tilde{\sigma}$ that includes both $\sigma^{*}$ and $\sigma$ as special cases. Let
\begin{align}
    PS_{\tilde{\sigma}}(w) := \E [\epsilon(\tilde{\sigma} (w,\epsilon))|w]
\end{align}
denote the expected current shock of the Markov policy $\tilde{\sigma}(w,\epsilon)$ conditionally on the current state $w$. Denote also by
\begin{align}
    \gamma^{\tilde{\sigma}}_0(w):= \E [\tilde{\sigma}(w,\epsilon)|w]
\end{align}
the conditional probability of choice $a=1$ under the strategy $\sigma$. The recursive equation for the value function takes the form:
\begin{align}
\label{eq:recursive:rho}
    \E \big[ \underbrace{V (w;\theta;\tilde{\sigma};\rho)}_{R_1(V(w;\tilde{\sigma};\rho))} &- \big( - R(1- \gamma_0^{\tilde{\sigma}}(w)) + (-\mu s)  \gamma_0^{\tilde{\sigma}}(w) + PS_{\tilde{\sigma}}(w) \\
    &+ \beta  \underbrace{\E  [V (w_{next};\theta;\tilde{\sigma};\rho) | w,\tilde{\sigma}(w,\epsilon)=0] (1- \gamma_0^{\tilde{\sigma}}(w))}_{R_3(V(w;\tilde{\sigma};\rho))} \nonumber  \\
    &+ \underbrace{ \beta \E_{e} [V (w_{next};\theta;\tilde{\sigma};\rho)|w,\tilde{\sigma}(w,\epsilon)=1 ] \gamma_0^{\tilde{\sigma}}(w) }_{R_4(V(w;\tilde{\sigma};\rho);\rho)} \big) \big] = 0  \nonumber
\end{align}
Recognize that the function $\rho$ appears in each of the expressions $$R_1(V(w;\tilde{\sigma};\rho)) ,R_3(V(w;\tilde{\sigma};\rho)),R_4(V(w;\tilde{\sigma};\rho);\rho)$$ inside the value function. In the expression $R_4(V(w;\tilde{\sigma};\rho);\rho)$, it also appears outside the value function, that is, through  the conditional distribution of $P(w_{next}|w,a=1)$. The derivative with respect to the outside presence of $\rho$ is equal to:
\begin{align*}
\mathbf{R}_{out}&:= \partial_{0} \E R_4(V(w;\tilde{\sigma};\rho);r(\rho-\rho_0)+\rho_0) \\
&=\partial_{0} \E_{e} V(r(\rho(w) - \rho_0(w)) + \rho_0(w) +e ;\theta; \tilde{\sigma}; \rho_0) \gamma_0^{\tilde{\sigma}}(w) \\
&= \E_{e} [\frac{d V(x;\theta;\tilde{\sigma};\rho_0)}{d x} |_{x=\rho_0(w)+e}] \gamma_0^{\tilde{\sigma}}(w) \\
&= \E [\frac{d V(x;\theta;\tilde{\sigma};\rho_0)}{d x} |_{x=w_{next}}|w,a=1] \gamma_0^{\tilde{\sigma}}(w).
\end{align*}
The derivative with respect to the inside presence of $\rho$ requires some preparation. For any Markov policy $\sigma$ I see that:
\begin{align}
\beta \E R_4(V (w_{next};\theta;\tilde{\sigma};\rho),\rho_0) &=
    \beta \E [V (w_{next};\theta;\tilde{\sigma};\rho)|w,\tilde{\sigma}(w,\epsilon)=1] \gamma_0^{\sigma}(w) \nonumber  \\
    &= \beta \E V (w_{next};\theta;\tilde{\sigma};\rho) \frac{1_{\tilde{\sigma}(w,\epsilon)=1}}{\gamma_0^{\tilde{\sigma}}(w)}  \gamma_0^{\tilde{\sigma}}(w) \nonumber \\
    &= \E V (w_{next};\theta;\tilde{\sigma};\rho) \lambda^{\tilde{\sigma}}_1(w_{next}) \nonumber \\
\label{eq:lambda1}
    &=  \E V (w;\theta;\tilde{\sigma};\rho) \lambda_1^{\tilde{\sigma}}(w),
\end{align}
where in the third line I used $   \lambda^{\tilde{\sigma}}_j(w_{next}) = \E [1_{\{ \tilde{\sigma}(w,\epsilon)=j\}}|w_{next}], j \in \mathcal{J},$ and in the fourth line I used the stationarity of the distribution. To sum up,
\begin{align}
\label{eq:lambda1}
    \beta \E R_4(V (w_{next};\theta;\tilde{\sigma};\rho),\rho_0) &=\E V (w;\theta;\tilde{\sigma};\rho) \lambda_1^{\tilde{\sigma}}(w).
\end{align}
The same argument can be applied to the term $\E R_3(V (w_{next};\theta;\tilde{\sigma};\rho))$:
\begin{align}
\label{eq:lambda0}
    \beta \E R_3(V (w_{next};\theta;\tilde{\sigma};\rho),\rho_0) &=\E V (w;\theta;\tilde{\sigma};\rho) \lambda_0^{\tilde{\sigma}}(w).
\end{align}
Combining  (\ref{eq:lambda1}) and (\ref{eq:lambda0}) gives:
\begin{align}
    \mathbf{R}_{ins}:= (1-\beta) \partial_0 \E V (w;\theta;\tilde{\sigma};r(\rho-\rho_0) + \rho_0)
\end{align}
for any Markov policy $\sigma$. According to the chain rule,
\begin{align*}
    \mathbf{R}_{ins} + \mathbf{R}_{out} =0.
\end{align*}
Therefore,
\begin{align*}
    \partial_0 \E V (w;\theta;\tilde{\sigma};r(\rho-\rho_0) + \rho_0)&= \frac{\beta}{1-\beta} \E_{e} \frac{d V(t;\theta;\tilde{\sigma};\rho_0)}{d t} |_{t=\rho_0(w)+e} \gamma_0^{\sigma}(w),
\end{align*}
and the bias correction term for the transition function is equal to:
\begin{align}
\label{eq:adj:rho}
    \alpha^{TRANS}_{\tilde{\sigma}}(D;\theta;\xi):=\frac{\beta}{1-\beta} \E_{e}[ \frac{d V(t;\theta;\tilde{\sigma};\rho_0)}{d t} |_{t=\rho_0(w)+e}] \gamma_0^{\tilde{\sigma}}(w)  ( s_{next} - \rho(w)),
\end{align}
where $\xi$ an unknown vector-valued function of the state variable $w$ whose true value $\xi_0 = \xi_0(\theta)$ consists of the original nuisance parameter $\eta_0$ and
\begin{align}
\label{eq:xi0trans}
    \xi_0(\theta):= \{ \eta_0, \Pi_0(\cdot,\theta) \}, 
\end{align}
where $\Pi_0(w,\theta):= \E [\frac{d V_{\tilde{\sigma}}(x;\theta;\eta_0)}{dx}|_{x = w_{next}} |w ]$ is the expectation of the derivative of the value function with respect to the state variable evaluated at the future state conditional on the current state $w$. To sum up, the bias correction term for Example \ref{ex:trans} is equal to:
\begin{align*}
    \alpha (D;\theta;\eta)&:= -\alpha^{CCP}(D;\theta;\gamma)+ \alpha^{TRANS}_{\sigma}(D;\theta;\gamma)- \alpha^{TRANS}_{\sigma^{*}}(D;\theta;\gamma),
\end{align*}
where $\alpha^{CCP}(D;\theta;\gamma)$ is given in
(\ref{eq:adj:gamma})  and $\alpha^{TRANS}_{\sigma}(D;\theta;\xi)$ is given in (\ref{eq:adj:rho}).
\end{proof}
\begin{proof} [Proof of Remark \ref{rm:Pug}]
\textbf{Step 1. Bias Correction Term for conditional choice probability of Apple. }
I seek for the bias correction term $\alpha^{CCP}(w;\gamma_A)$ such that takes the form:
\begin{align*}
   \alpha^{CCP}(w;\theta;\gamma_A) &= \Gamma_{A} (w;\theta) (1_{a_{A}=1} - \gamma_{A}(w))
\end{align*}
and the function $\Gamma_{A} (w;\theta)$ satisfies the condition:
\begin{align*}
    \partial_0 \E V(w;\theta; \sigma^{*};r(\gamma_{A}(w) - \gamma_{A,0}(w))+ \gamma_{A,0}(w)) &= \E \Gamma_{A} (w;\theta) [\gamma_{A}(w) - \gamma_{A,0}(w)].
\end{align*}
I find the function $ \Gamma_{A} (w;\theta)$ by the application of the implicit function theorem. I write $V(w;\theta;\sigma^{*};\eta)=V(w;\theta;\sigma^{*};\gamma_A)$ in order to focus on the nuisance parameter $\gamma_A$ holding the other parameter $\gamma_{P}$ at its true value. The recursive equation of the value function of Apple is:
\begin{align*}
    \E V(w;\theta;\sigma^{*};\gamma_A) &= \E [-R (1-\gamma_A(w)) + (-\mu s)\gamma_A(w) + PS_{\sigma^{*}}(\gamma_A) \\
    &+ \delta_0 \gamma_{P,0}(w)  + (\delta_1 - \delta_0) \gamma_{P,0}(w)   \gamma_A(w) \\
    &+ \beta \E [V(w_{next};\theta;\sigma^{*};\gamma_A)|w,a=0] (1-\gamma_A(w)) \\
    &+ \beta \E [V(w_{next};\theta;\sigma^{*};\gamma_A)|w,a=1] (\gamma_A(w))].
\end{align*}
The function $\gamma_A$ appears in each summand of the equation above both inside and outside the value function $V(w;\theta;\sigma^{*};\gamma_A)$.  Let $\mathbf{S}_{out},\mathbf{S}_{ins}$ be defined as in Example \ref{ex:trans}. Taking the derivative with repsect to the oustide presence of $\gamma_A$, I get:
\begin{align*}
    \partial_{0} \mathbf{S}_{out} &=  -(-\mu s +R) + (\delta_1 - \delta_0) \gamma_{P,0}(w) + \frac{d PS_{\sigma^{*}}(\gamma) }{d\gamma}|_{\gamma_{A,0}} \\
    &+ \beta \E [V(w_{next};\theta;\sigma^{*};\gamma_{A,0})|w,a=1] -\beta \E [V(w_{next};\theta;\sigma^{*};\gamma_{A,0})|w,a=0]
\end{align*}
where the second line follows from the definition of the choice-specific value functions and the third line follows from the logistic assumption of the shocks by (\ref{eq:HotzMiller}).Taking the derivative with respect to the inside presence of $\gamma_A$, I get:
\begin{align*}
     \partial_{0} \mathbf{S}_{ins} &= (1-\beta) \partial_0 \E V (w;\theta;\sigma^{*};r(\gamma_{A}-\gamma_{A,0}) + \gamma_{A,0}).
\end{align*}
The function $\Gamma_A(w;\theta)$ follows:
\begin{align*}
    \Gamma_A(w;\theta) &= (-\mu s +R) + (\delta_1 - \delta_0) \gamma_{P,0}(w) + \frac{d PS_{\sigma^{*}}(t) }{dt}|_{t=\gamma_{A,0}} \\
    &+ \beta \E [V(w_{next};\theta;\sigma^{*};\gamma_A)|w,a=1] -\beta \E [V(w_{next};\theta;\sigma^{*};\gamma_A)|w,a=0]
\end{align*}
and \begin{align*}
    \alpha^{CCP}(D;\theta;\gamma_A) = \frac{1}{1-\beta} \Gamma_A(w;\theta) (1_{a_{A}=1} - \gamma_A(w)).
\end{align*}
I write $V(w;\theta;\sigma;\eta)=V(w;\theta;\sigma;\gamma_{P})$ holding $\gamma_{A}$ fixed at its true value $\gamma_{A,0}$.
I find the bias correction term  $ \alpha^{*}_{P}(D;\theta;\gamma_{P})$ such that:
\begin{align*}
   \partial_{\gamma_{P}} \E [V (w;\theta;\sigma^{*};\gamma_P) +\alpha^{*}_{P}(D;\theta;\gamma_{P})]   &=0.
\end{align*}
The bias correction term takes the form:
\begin{align*}
    \alpha^{*}_{P}(D;\theta;\gamma_{P}):= \Gamma_{P}(w;\theta) (1_{\{ a_{P}=1\}}-\gamma_{P}(w)).
\end{align*}
The function $\Gamma_{P}(w;\theta) $ is defined by the application of the  implicit function theorem. Using the approach of Example \ref{ex:trans}, I find that
\begin{align*}
    \partial_0 \E V (w;\theta;\sigma^{*};r(\gamma_{P}  - \gamma_{P,0}) + \gamma_{P,0}) (1-\beta)&= \delta_0 + (\delta_1 - \delta_0) \gamma_{A,0}(w).
\end{align*}
and
\begin{align*}
    \Gamma_{P}(w;\theta) &=\frac{1}{1-\beta}( \delta_0 + (\delta_1 - \delta_0) \gamma_{A,0}(w)).
\end{align*}
Similarly,
 the bias correction term  $ \alpha^{\sigma}_{P}(D;\theta;\gamma_{P})$  takes the form:
\begin{align*}
    \alpha^{\sigma}_{P}(D;\theta;\gamma_{P}):= \Gamma^{\sigma}_{P}(w;\theta) (1_{\{ a_{P}=1\}}-\gamma_{P}(w)),
\end{align*}
where
\begin{align*}
    \Gamma^{\sigma}_P(w;\theta) = \frac{1}{1-\beta}( \delta_0 + (\delta_1 - \delta_0) \gamma_{0}^{\sigma}(w)),
\end{align*}
where $\gamma^{\sigma}(w):= \Pr(\sigma(w,\epsilon)=1|w)$.
To sum up, the bias correction term for the conditional probability of Player 2 entry is equal to:
\begin{align}
\label{eq:adj:final:Pug}
      \alpha^{*}_{P}(D;\theta;\gamma_{P}):=\frac{1}{1-\beta} (\delta_1 - \delta_0) (\gamma_{0}^{\sigma}(w) - \gamma_{A,0}(w)) (1_{\{a_{P} =1\}} -\gamma_P(w)).
\end{align}
\end{proof}

\begin{remark}[Bias correction term for point-identified problems]
\label{rm:point}
Suppose the parameter of interest $\beta_0$ is the population average of the value function $$\beta_0:=\E [V(w;\theta_0;\sigma^{*};\eta)].$$  Then the bias correction term for the conditional choice probability of player one is equal to:
\begin{align*}
	\alpha^{CCP}(D;\theta;\gamma):= - \frac{1}{1-\beta} \frac{2}{\gamma(w)} (1_{\{ a_1 =  1 \}} - \gamma(w) ),
\end{align*}
where the choice probability $\gamma_0(w)$  is defined in (\ref{eq:HotzMiller}). Furthermore, if the parameter interest $\beta_0$ is equal to the weighted average of the value function 
\begin{align}
\label{eq:weightedpoint}
	 \beta_0 = \E q(w)  V(w;\theta_0;\sigma^{*};\eta), 
\end{align}
the bias correction term is equal to:
\begin{align*}
	\alpha^{CCP}_q(D;\theta;\gamma) = -\frac{ q(w)}{ q(w)  - \beta \lambda(w)} \frac{2}{\gamma(w)} (1_{\{ a_1 =  1 \} }- \gamma(w) ),
\end{align*}
where $\lambda(w'):= \E [q(w)| w_{next} = w']$ is the expectation of the weighting function $q(w)$ conditionally on the future state $w_{next}$.
\end{remark}

\section{Proofs from Section \ref{sec2:dynamic}}

\begin{proof}[Proof of Lemma \ref{lem:mylove}]
Let $\eta$ be as in (\ref{eq:eta}). Consider the value function $V(w;\theta;\sigma^{*};\eta) = V(w;\theta;\sigma^{*};\gamma_j)$ as a function of $\gamma_j$ holding the other components  of the nuisance parameter $\eta$ fixed at the true value.
The recursive equation for the value function takes the form:
\begin{align}
\label{eq:recursive:gammaj}
    \E q(w) \underbrace{V(w;\theta;\sigma^{*};\gamma_j) }_{S_1(V(w;\sigma^{*};\gamma_j))}&- \underbrace{\sum_{j=2}^J \E_{\epsilon_{-1}} q(w) [\tilde{\pi} ((j,\sigma^{*}_{-1}(w,\epsilon_{-1}));w;\theta_0)}_{}\\
     &- \underbrace{ \tilde{\pi} ((1,\sigma^{*}_{-1}(w,\epsilon_{-1}));w;\theta_0)) |w]\gamma_j(w) + PS_{\sigma_1^{*}}(\gamma(w))}_{S_2(\gamma_j)} \\
    &+\underbrace{\beta \E q(w)[ V(w_{next};\theta;\sigma^{*};\gamma_j)|w,(1,\sigma^{*}_{-1} (w,\epsilon_{-1}))]}_{S_3(V(w;\sigma^{*};\gamma_j))} \nonumber  \\
     &+\beta \sum_{j=2}^J  \E_{\epsilon_{-1}} [ V(w_{next};\theta;\sigma^{*};\gamma_j)|w,(j,\sigma^{*}_{-1}(w,\epsilon_{-1}))]\\
    &- \E_{\epsilon_{-1}} [ V(w_{next};\theta;\sigma^{*};\gamma_j)|w,(1,\sigma^{*}(w;\epsilon_{-1}))]   \gamma_j(w) =0. \nonumber
\end{align}
The function $\gamma_j(w)$ appears in the equation both inside the value function $V(w;\theta;\sigma^{*};\gamma_j)$
and outside of  it. Taking the derivative with respect to the outside presence of $\gamma_j$ gives:
\begin{align*}
    \mathbf{S}_{out} &:= \E [\tilde{\pi} ((j,\sigma^{*}(w;\epsilon_{-1}));w;\theta_0) -  \tilde{\pi} ((1,\sigma^{*}(w;\epsilon_{-1}));w;\theta_0)) |w] + \partial_{\gamma_{j,0}} PS_{\sigma^{*}}(\gamma) \\
    &+ \beta (\E_{\epsilon_{-1}} [ V(w_{next};\theta;\sigma^{*};\gamma_j)|w,(j,\sigma^{*}(w;\epsilon_{-1}))] - \E_{\epsilon_{-1}} [ V(w_{next};\theta;\sigma^{*};\gamma_j)|w,(1,\sigma^{*}(w;\epsilon_{-1}))] ).
\end{align*}
Taking the derivative with respect to the inside presence gives:
\begin{align*}
     \mathbf{S}_{ins} := (q(w)-\beta \lambda(w)) \partial_0 \E V (w;\theta;\sigma^{*};r(\gamma_j - \gamma_{j,0}) +\gamma_{j,0})
\end{align*}
using the same argument as in (\ref{eq:trans:change}). Therefore, the bias correction term is
\begin{align*}
    \alpha^{CCP}_{j}(D;\theta;\gamma_j) :=
    -\frac{q(w)}{q(w) - \beta \lambda(w)} \Gamma_{\gamma,j}(w;\theta) \left(1_{\{a_1=j\}} - \gamma_j(w)\right),
\end{align*}
where $\Gamma_{\gamma,j}(w,\theta)$ is given by
\begin{align*}
   \Gamma_{\gamma,j}(w,\theta)&:= \E_{\epsilon_{-1}}[\tilde{\pi} ((j,\sigma^{*}(w;\epsilon_{-1}));w;\theta_0) -  \tilde{\pi} ((1,\sigma^{*}(w;\epsilon_{-1}));w;\theta_0)) |w] + \partial_{\gamma_{j,0}} PS_{\sigma^{*}}(\gamma) \\
    &+ \beta (\E_{\epsilon_{-1}} [ V(w_{next};\theta;\sigma^{*};\gamma_j)|w,(j,\sigma^{*}(w;\epsilon_{-1}))] \\
    &-\E_{\epsilon_{-1}} [ V(w_{next};\theta;\sigma^{*};\gamma_j)|w,(1,\sigma^{*}(w;\epsilon_{-1}))] ).
\end{align*}
\end{proof}

\begin{proof}[Proof of Lemma \ref{lem:trans}]
Let $a$ be a  profile of actions. Let $Q(u,w,a)$ be the conditional quantile function of the conditional distribution $\Pr (w_{next}|w,a)$ defined as
\begin{align*}
\Pr (w_{next} \leq Q(u,w,a)|w,a) =u, \quad u \in [0,1]. 
\end{align*}
Thus I treat $\rho(u,w,a):=Q(u,w,a)$ as the nuisance parameter instead of $\Pr (w_{next}|w,a)$. Let $\tilde{\sigma} \in \{ \sigma, \sigma^{*}\}$ be a Markov policy. Consider the value function $V(w;\theta;\tilde{\sigma} ;\eta) = V(w;\theta;\tilde{\sigma};\rho)$ as a function of $\rho(u,w,a)$ holding the other components  of the nuisance parameter $\eta$ fixed at the true value. By the property of the conditional quantile function, the bias correction term $\alpha_{\sigma}^{TRANS}(D;\theta;\eta)$ is equal to:
\begin{align*}
	\alpha_{\tilde{\sigma}}^{TRANS}(D;\theta;\eta) &= \Pi(w;\theta) \cdot \frac{ 1_{\{ w_{next} \leq  \rho(u,w,a) \}} - u}{f(w_{next}|w,a)},
\end{align*}
where $f(w_{next}|w,a)$ is the conditional density of the future state $w_{next}$ given the current state and  the function $ \Pi(w;\theta)$ is defined by the following equation:
\begin{align*}
	\partial_0 \E q(w) V(w;\theta;\tilde{\sigma}; r(\rho - \rho_0)+ \rho_0) &= \E \Pi(w;\theta) [\rho - \rho_0].
\end{align*}
I find the function $ \Pi(w;\theta)$ by the application of the implicit function theorem to the recursive equation (\ref{eq:recursive}). The equation takes the form:
\begin{align*}
\E q(w) \bigg[  V(w;\theta;\tilde{\sigma} ;\rho) &- \sum_{j=1}^J \E_{\epsilon_{-1}} q(w) [\tilde{\pi} ((j,\sigma^{*}_{-1}(w,\epsilon_{-1}));w;\theta_0)] \prod_{k=1}^K \Pr (\sigma_k(w,\epsilon_k)=a_k|w) \\
&+ PS_{\tilde{\sigma}_1}(w) + \beta \sum_{a' \in \prod_{k=1}^K \mathcal{A}_k } \E [V (w_{next};\theta;\tilde{\sigma};\eta_0) | w,a'] \Pr (\tilde{\sigma}_k(w,\epsilon_k)=a'_k|w)   \bigg] \\
&= 0.
\end{align*}
The nuisance parameter $\rho$ enters the equation above both inside and outside of the value function. The derivative with respect to the outside presence of $\rho$ takes the form:
\begin{align*}
& \beta \partial_0 \E  q(w) [V ( r( \rho(u_{next} ,w,a) - \rho_0(u_{next},w,a)) + \rho_0(u_{next},w,a)   ;\theta;\tilde{\sigma};\eta_0) | w,a] \prod_{k=1}^K \Pr (\tilde{\sigma}_k(w,\epsilon_k)=a_k|w) \\
&= \beta \E  q(w)  \nabla_{w_{next}} V(w_{next};\theta;\tilde{\sigma};\eta_0) [ \rho(u_{next},w,a) - \rho_0(u_{next},w,a)],
\end{align*}
where $u_{next}$ is the random draw from the $U[0,1]$. The derivative with respect to the inside presence is derived using the same argument of Lemma \ref{lem:mylove} that applies regardless  of the  equilibrium property of the strategy $\sigma$. The argument of Lemma \ref{lem:trans} yields the bias correction term:
\begin{align*}
	\alpha_{\sigma}^{TRANS}(D;\theta;\eta) &= \frac{\beta q(w)}{q(w) - \beta \lambda(w)} \E [\nabla_{w_{next}} V(w_{next};\theta;\tilde{\sigma};\eta_0) |w,a]  \cdot \frac{ 1_{\{ w_{next} \leq  \rho(u,w,a) \}} - u}{f(w_{next}|w,a)}.
\end{align*}

\end{proof}

\begin{proof}[Proof of Lemma \ref{lem:opponent}]
Fix an action $j_2$ of player 2. Let $\gamma_{j_2}:= \Pr (\sigma_2^{*}(w,\epsilon_2) = j_2|w)$ be the conditional probability of the choice $j_2$ under the equilibrium strategy $\sigma_2^{*}(w,\epsilon_2)$ of player 2. I consider the value function $V(w;\theta;\tilde{\sigma};\eta) = V(w;\theta;\tilde{\sigma};\gamma_{j_2})$ holding the other nuisance parameters fixed at their true values. The bias correction term takes for $\gamma_{j,2}$ takes the form:
\begin{align*}
	\alpha_{j_2}^{CCP,op} (D;\theta;\gamma_{j_2}) = \Gamma_{op}(w;\theta) (1_{\{ a_2 = j_2 \}} - \gamma_{j_2}(w)),
\end{align*}
where the function $\Gamma_{op}(w;\theta) $ is defined by the following equation:
\begin{align*}
	\partial_0 \E q(w) V(w;\theta;\tilde{\sigma}; r(\gamma_{j_2} - \gamma_{j_2,0} )+ \gamma_{j_2,0} ) &= \E \Gamma_{op}(w;\theta)[\gamma_{j_2} (w) - \gamma_{j_2,0} (w)].
	\end{align*}
\end{proof}
I find the function $ \Gamma_{op}(w;\theta)$ by the application of the implicit function theorem to the recursive equation  (\ref{eq:recursive}). In the case $K=2$ of two players  (\ref{eq:recursive}) takes the form:
\begin{align*}
\E q(w)  \bigg[  V(w;\theta;\tilde{\sigma} ;\rho) &- \sum_{j_1=1}^{A_1} \sum_{j_2=1}^{A_2}  \tilde{\pi} ((j_1,j_2),w,\theta) \Pr (\tilde{\sigma}(w,\epsilon_1)=j_1|w)  \gamma_{j_2} (w) \\
&+ PS_{\tilde{\sigma}_1}(w) + \beta   \sum_{j_1=1}^{A_1} \sum_{j_2=1}^{A_2}    \E [V (w_{next};\theta;\tilde{\sigma};\eta_0) | w, (j_1,j_2)]  \Pr (\tilde{\sigma}(w,\epsilon_1)=j_1|w) \gamma_{j_2} (w)  \bigg] \\
&= 0.
\end{align*}
The nuisance parameter $\gamma_{j_2}(w)$ enters the equation above both inside and outside of the utility function. The function $ \gamma_{j_2} (w) $ enters linearly in its outside presence. Therefore, the derivative with respect to the outside presence takes the form:
\begin{align*}
	&\E q(w)  \bigg[   \sum_{j_1=1}^{A_1} \big( \tilde{\pi} ((j_1,j_2);w;\theta) -  \tilde{\pi} ((j_1,1);w;\theta) \big) \Pr (\tilde{\sigma}(w,\epsilon_1)=j_1|w) + \\
	&\beta   \sum_{j_1=1}^{A_1}  \big (  \E [V(w_{next};\theta;\tilde{\sigma} ) | w, (j_1,j_2)] - \E [V(w_{next};\theta;\tilde{\sigma} ) | w, (j_1,1)] \big) \Pr (\tilde{\sigma}(w,\epsilon_1)=j_1|w)     \bigg] \cdot \\
	&\cdot [\gamma_{j,2} (w) - \gamma_{j,2,0} (w)].
\end{align*}
The derivative with respect to the inside presence is derived using the same argument of Lemma \ref{lem:mylove} that applies regardless  of the  equilibrium property of the strategy $\tilde{\sigma}$ and yields the bias correction term:
\begin{align*}
	\alpha^{CCP,op}_{j_2} (D;\theta;\gamma_{j_2}) &= \frac{ q(w)}{q(w) - \beta \lambda(w)}  \bigg[   \sum_{j_1=1}^{A_1} \big( \tilde{\pi} ((j_1,j_2);w;\theta) -  \tilde{\pi} ((j_1,1);w;\theta) \big) \Pr (\sigma_1(w,\epsilon_1)=j_1|w) + \\
	&\beta   \sum_{j_1=1}^{A_1}  \big (  \E [V(w_{next};\theta;\tilde{\sigma} ) | w, (j_1,j_2)] - \E [V(w_{next};\theta;\tilde{\sigma} ) | w, (j_1,1)] \big) \Pr (\sigma_1(w,\epsilon_1)=j_1|w)     \bigg] \cdot \\
	&\cdot  (1_{\{ a_2 = j_2 \}} - \gamma_{j_2}(w)).
	\end{align*}

\begin{proof}[Proof of Lemma \ref{lem:linearity}]
Suppose there exists a function $B(a,w,\zeta)$ that is the basis function of the per-period utility function (\ref{eq:perperiod}):
\begin{align*}
\tilde{\pi}(a;w;\theta;\zeta) = \theta \cdot B(a;w;\zeta).
\end{align*}
Since $\zeta$ is an identified parameter of the distribution $P_D$, I drop it from the notation. The value function is a linear function of $\theta$:
\begin{align*}
V(w;\theta;\sigma;\eta) &= \underbrace{\E [\sum_{t \geq 0} \beta^t  B( \sigma(w_t,\epsilon_t);w_t) | w]}_{=:\Psi_1(w,\sigma)} \cdot \theta + \underbrace{\E [\sum_{t \geq 0} \beta^t \epsilon_{1,t}( \sigma_1(w_t,\epsilon_t)) | w]}_{=:\Psi_2(w,\sigma)},
\end{align*}
where $\Psi_1(w,\sigma)$ and  $\Psi_2(w,\sigma)$ are equal to the expected discounted value of the deterministic per-period utility and the private shock, respectively. Therefore, $V(w;a;\sigma;\eta)$ is a linear function of $\theta$.  The profile-specific value function 
\begin{align*}
\E [V(w_{next};\theta;\sigma;\eta) | w, a] &= \E [\sum_{t \geq 1} \beta^t  B( \sigma(w_t,\epsilon_t);w_t) | w_0 =w,a_0=a] \cdot \theta \\
&+ \E [\sum_{t \geq 1} \beta^t \epsilon_{1,t}( \sigma_1(w_t,\epsilon_t)) | w_0=w,a_0=a]
\end{align*}
is also a linear function of $\theta$.  Therefore, the bias correction terms (\ref{eq:alphajccp}) and  (\ref{eq:opponent}) are linear functions of $\theta$.
The expected derivative of the value function is also a linear function of $\theta$:
\begin{align*}
 \E [ \nabla_{w_{next}} V(w_{next};\theta;\sigma;\eta) |  w_0=w,a ] &= \E [ \partial_{w_{next}} \Psi_1(w_{next},\sigma) |w, a] \cdot \theta +  \E [ \partial_{w_{next}} \Psi_2(w_{next},\sigma) |w, a].
\end{align*}
Therefore, the bias correction term (\ref{eq:trans}) is a linear function of $\theta$.
\end{proof}

\section{Proofs from Section \ref{sec2:theory}}

\label{appendix:b}
\subsection{Auxiliary Lemmas}
\begin{lemma}[Basic Inequalities]
\label{lem:basicineq}
The following  inequalities hold for all vectors $x, y \in \mathcal{R}^d$:
 \begin{align*}
| \| x \|^2 - \|y \|^2 | &\leq \| x - y \| \| x + y \|\\
|\| x \|_{+}^2 - \|y \|_{+}^2 | &\leq \| x - y \|_{+}^2 + 2\|x-y \|_{+}  \|  y \|_{+}.
\end{align*}
\end{lemma}
\begin{proof}
The first inequality holds:
 \begin{align*}
| \| x \|^2 - \|y \|^2 | &=| \sum_{j=1}^d (x_j^2 - y_j^2) | \leq \sum_{j=1}^d | x_j - y_j || x_j + y_j |  \leq  \| x - y \| \| x + y \|.
   \end{align*}
   The second inequality holds:
    \begin{align*}
  &| \| x \|_{+}^2 - \|y \|_{+}^2 | =| \sum_{j=1}^d ((x^{+})_j^2 - (y^{+})_j^2) |  =  | \sum_{j=1}^d (x^{+}_j - y^{+}_j) (x^{+}_j + y^{+}_j)  | \\
  &\leq  \sum_{j=1}^d |(x^{+}_j - y^{+}_j)|| (x^{+}_j + y^{+}_j) = \sum_{j=1}^d |(x^{+}_j - y^{+}_j)|^2 + 2 \sum_{j=1}^d |(x^{+}_j - y^{+}_j)| y^{+}_j\\
  &\leq \sum_{j=1}^d ( |x_j - y_j|^2_{+} + 2|x_j - y_j|_{+}  y^{+}_j)   \leq \| x-y \|_{+}^2 + 2 \| x- y \|_{+} \|y\|_{+}.
   \end{align*}
\end{proof}

The assumptions below are the high-level assumptions on the population criterion function $Q(\theta,\xi)$. I will prove Theorem \ref{thm:main:ineq:chap2} under the assumptions below. I verify these assumptions from the conditions of Section \ref{sec2:theory} in the Subsection \ref{subsec:proof:aux}. 

The proofs below are defined for the simple sample splitting procedure. The sample $(D_i)_{i=1}^N$ is decomposed into the auxiliary sample $J_1$ and the main sample $J_2$ of size $n:=[N/2]$ each. Let the nuisance parameter $\widehat{\xi}$ be constructed on the sample $J_1$. For each $\xi \in \Xi$ define the sample criterion function as
\begin{align}
\label{eq:qn:appendix}
	Q_n(\theta, \xi):= \| \frac{1}{n} \sum_{i \in J_2}  g(D_i, \theta, \xi) \widehat{W}(\theta), \|_{+}^2.
\end{align}
where $\widehat{W}(\theta)$ is a diagonal weighting matrix that converges to a  diagonal weighting matrix $W(\theta)$ with strictly positive diagonal elements.

\begin{assumption}[Consistency of $Q_n(\theta,\xi_0)$]
\label{ass:cons}
(1) Regularity of $Q(\theta, \xi_0)$. The function $Q(\theta, \xi_0)$ is a non-negative lower semicontinuous function such that $\forall \epsilon > 0, \exists \delta(\epsilon)>0$ such that
\begin{align}
\label{eq:cons1}
\inf_{\Theta \setminus {\Theta_I}^{\epsilon}} Q(\theta, \xi_0) \geq \delta(\epsilon)>0.
\end{align}
(2) Fast Convergence on $\Theta_I$. The sample loss $Q_n(\theta,\xi_0)$ converges to the population loss $Q(\theta,\xi_0)$ uniformly over the identified set $\Theta_I$ at the  rate $n$:
  \begin{align} \label{eq:cons2} \sup_{\Theta_I} n Q_n (\theta,\xi_0) = O_{P} (1). \end{align}
(3) Slow Convergence on $\Theta$.  The sample loss $Q_n(\theta,\xi_0)$ converges to the population loss $Q(\theta,\xi_0)$ uniformly over the whole set $\Theta$ in the semi-metric $\| \cdot \|_{+} $ at the rate $\sqrt{n}$:  \begin{align} \label{eq:cons3} \sup_{\Theta} (Q (\theta, \xi_0) - Q_n(\theta, \xi_0))_{+} = O_{P} (1/\sqrt{n}). \end{align}
(4) There exist positive constants $\delta, \kappa>0$ such that for any $p \in (0,1)$  there exist constants $d_p, n_{p}$ such that for any $n \geq n_p$:
\begin{align} \label{eq:cons4} \inf_{\xi(\cdot) \in \Xi_n} \Pr ( Q_n(\theta,\xi) \geq \kappa [ d(\theta, \Theta_I) \wedge \delta]^{2} \quad \forall \theta: d_H(\theta,\Theta_I)\geq (d_p/n)^{1/2}) \geq 1 - p, \end{align} where $\Xi_n:= \bigtimes_{\theta \in \Theta} \Xi^{\theta}_n$, and $\Xi^{\theta}_n$ is defined in Condition \ref{ass:smallbiasvalue}.
\end{assumption}

\begin{assumption}[No effect of the first stage estimation error]
\label{ass:noeffect}
(1) Slow Convergence on $\Theta$. For any $p>0$ there exist constants $r_p,n_p$ such that $\forall n \geq n_p$ the  difference of the sample losses $Q_n(\theta, \xi) -Q_n(\theta, \xi_0)$ evaluated at the nuisance value $\xi(\theta)$ and at the true value $\xi_0(\theta)$ converges uniformly  over $\Theta$  at rate $\sqrt{n}$ for any element $\xi(\theta) \in \Xi_n^{\theta}$: $$ \inf_{\xi(\cdot) \in \Xi_n} \Pr ( \sqrt{n} \sup_{\Theta}   |Q_n(\theta, \xi) -Q_n(\theta, \xi_0) | \leq r_p) \geq 1-p .$$
(2) Fast Convergence on $\Theta_I$. For any $p, \epsilon>0$ there exists $n_{p, \epsilon}$ such that $\forall n \geq n_{p, \epsilon}$ the  difference of the sample losses $Q_n(\theta, \xi) -Q_n(\theta, \xi_0)$ converges uniformly  over $\Theta_I$   for any element $\xi(\cdot) \in \Xi_n$: $$ \inf_{\xi(\cdot) \in \Xi_n} \Pr ( n \sup_{\Theta_I}   |Q_n(\theta, \xi) -Q_n(\theta, \xi_0) | \leq \epsilon) \geq 1-p .$$
    \end{assumption}

\begin{assumption}[Fast Convergence on the $\epsilon_n$- Expansion of the Identified Set]
\label{ass:fast}
Let $\epsilon_n  = O_{P} ((\frac{ 1 \vee \widehat{c}}{n} )^{1/2}) $ be the convergence rate in Theorem \ref{lem:rate}. For any $p, \epsilon>0$ there exists $n_{p, \epsilon}$ such that $\forall n \geq n_{p, \epsilon}$ the  difference between the sample losses $Q_n(\theta, \xi) -Q_n(\theta, \xi_0)$ converges uniformly  over $\Theta^{\epsilon_n}_I$   for any element $\xi \in \Xi_n$: $$ \inf_{\xi \in \Xi_n} \Pr ( n \sup_{\Theta^{\epsilon_n}_I}   |Q_n(\theta, \xi) -Q_n(\theta, \xi_0) | \leq \epsilon) \geq 1-p .$$
\end{assumption}

\begin{lemma}[Coverage,Consistency, and Rate of Convergence for Loss Functions]
\label{lem:rate}
Let $\widehat{c}/n \rightarrow_p 0$ and $\inf_{\xi \in \Xi_n} \Pr(\sup_{\theta \in \Theta_I} n Q_n (\theta, \xi) \leq \widehat{c} ) = 1-o(1)$ (containment  $\Theta_I \subseteq \widehat{ \Theta}_I$).
Then, Assumptions \ref{ass:cons}[1-3] and \ref{ass:noeffect}[1]   imply that $ \Theta_I \subseteq  \widehat{\Theta}_I $ w.p. 1 and $d_H(\widehat{\Theta}_I, \Theta_I) = o_{P} (1)$. Assumptions \ref{ass:cons}[1-4] and \ref{ass:noeffect}[1] imply that $d_H(\widehat{\Theta}_I, \Theta_I) = O_{P} ((1 \vee \widehat{c} )/n )^{1/2}$. In the case $\Theta_I = \Theta$,  Assumption \ref{ass:cons} implies that  $d_H(\widehat{\Theta}_I, \Theta_I) =0$. In the case Assumption \ref{ass:noeffect}[2] holds, $\widehat{c}$ can be chosen as $\widehat{c} = O_{P}(1)$.
\end{lemma}

\begin{lemma}[Concentration of Estimated Moments]
\label{lem:concentrate}
W.p. $\rightarrow 1$,
\begin{align}
\label{eq:moment2}
\sup_{\theta \in \Theta } |\Gn  [g(D_i, \theta, \widehat{\xi}) - g(D_i, \theta, \xi_0) ]| \lesssim_{P} r_n' \log (1/r_n') + n^{-1/2+1/s} \log n.
\end{align}
\end{lemma}

\begin{lemma}[Sufficient Conditions for General Moment Problems]
\label{lem:rateineq}
Conditions \ref{ass:smallbiasvalue}, \ref{ass:concentration:chap2}, (\ref{eq:pdonsker}) imply that Assumptions  \ref{ass:noeffect} and  \ref{ass:fast}  are satisfied  for  the moment inequalities problem with the population loss
$Q(\theta, \xi_0) = \| \Ep g(W_i, \theta, \xi_0(\theta)) \|_{+}$
and its sample analog $Q_n(\theta, \widehat{\xi}) = \| \En g(D_i, \theta, \widehat{\xi}) \|_{+}$,
as well as the moment equalities problem with the population loss $Q(\theta, \xi_0) = \| \Ep g(D_i, \theta,\xi_0(\theta)) \|$
and its sample analog $Q_n(\theta, \widehat{\xi}) = \| \En g(W_i, \theta,\widehat{\xi}) \|$.
\end{lemma}

\begin{assumption}[Limit Distribution of $C_{n}$ ]
\label{ass:c4}
There exists a law ${\cal C}$ such that $$\Pr (C_{n} \leq c )  \rightarrow \Pr ({\cal C} \leq c ) \quad \forall c \in \mathcal{R},$$ where the distribution function of 
${\cal C}$  is continuous on $[0, \infty)$.

\end{assumption}

\begin{assumption}[Approximability of $C_{n}$]
\label{ass:c5}
For any  sequence of random measurable sets $\Theta_n$ such that $d_H(\Theta_n, \Theta_I)  = o_{P} (n^{-1/2})$, the sequence of the suprema of the sample loss over sets $\Theta_n$, $C'_{n} := \sup_{\theta \in \Theta_n} n Q_n(\theta, \xi_0)$ satisfies uniform convergence: $$\sup_{c \in R} | \Pr (C'_{n} \leq c)   - \Pr ({\cal C} \leq c ) | = o(1).   $$
\end{assumption}

\subsection{Proof of Theorems from Main Text}

\begin{proof}[Proof of Theorem \ref{thm:main:ineq:chap2} ]. Assumption \ref{ass:cons}[1] holds:
\begin{align*}
    \inf_{\Theta \setminus \Theta_I^{\epsilon}} Q(\theta,\xi_0) &= \| \Ep g(D_i, \theta, \xi_0) \|_{+}^2 \geq C_{\min}^2 ( \epsilon \wedge \delta_{\min})^2 > 0  \tag{(\ref{ass:partid})}.
\end{align*}
Assumption \ref{ass:cons}[2] holds by Lemma \ref{lem:basicineq} and the  $P$-Donsker property of $g(D_i, \theta,\xi_0(\theta))$:
\begin{align*}
\sup_{\theta \in \Theta_I} n  \| \En g(D_i, \theta, \xi_0) \|_{+}^2 &\leq \sup_{\theta \in \Theta_I}  \| \Gn  g(D_i, \theta, \xi_0) \|_{+}^2 \\
&+ 2\| \Gn  g(D_i, \theta, \xi_0)  \|_{+} \underbrace{\| \sqrt{n} \E g(D_i, \theta, \xi_0) \|_{+} }_{=0\text{ by Definition of $\Theta_I$}}\\
&+ \underbrace{\| \sqrt{n} \E g(D_i, \theta, \xi_0) \|^2_{+} }_{=0\text{ by Definition of $\Theta_I$}}\\
&\Rightarrow^d \| \Delta(\theta ) \|_{+}^2  = O_{P} (1).
\end{align*}
Assumption \ref{ass:cons}[3] holds by Lemma \ref{lem:basicineq} and the  $P$-Donsker property of $g(D_i, \theta,\xi_0)$:
\begin{align*}
\sup_{\theta \in \Theta} \sqrt{n} | \| \Ep g(D_i, \theta, \xi_0) \|_{+}^2 -   \| \En g(D_i, \theta, \xi_0) \|_{+}^2 |_{+}= O_{P} (1).
\end{align*}
Assumption \ref{ass:cons}[4]   holds. Fix an amount of probability $p>0$ and let $n$ be large enough so that minimal eigenvalue $\widehat{W}$ is bounded below w.p.$1-p/4$:
$$\Pr (\mathcal{W}_n) := \Pr (\{\lambda_{\max} \geq \max \eig \quad \widehat{W}(\theta) \geq \min \eig \quad \widehat{W}(\theta) \geq \lambda_{\min} \}) \geq 1-p/4.$$ By Condition \ref{ass:smallbiasvalue}(1) $$\Pr (\mathcal{B}_n)= \Pr (\widehat{\xi}(\cdot) \in \Xi_n) \geq 1-p/4,$$ by Lemma \ref{lem:concentrate}  $$ \Pr ( \mathcal{F}_n):= \Pr (\sup_{\theta \in \Theta} | \Gn g(D_i, \theta, \widehat{\xi}) -  g(D_i, \theta, \xi_0) | \leq  C_{p/4}) \geq 1-p/4,  $$ and  by (\ref{eq:pdonsker}),  $$ \Pr ( \mathcal{G}_n):= \Pr (\| \Gn g(D_i, \theta,\xi_0) \| \leq C_{p/4}) \geq 1-p/4.$$

On the event $\mathcal{B}_n \cap \mathcal{W}_n \cap \mathcal{F}_n \cap \mathcal{G}_n$ for $n$ sufficiently large, $  n Q_n(\theta, \widehat{\xi})$ is bounded below:
\begin{align*}
&n Q_n(\theta, \widehat{\xi}) = \| \left(\Gn g(D_i, \theta, \widehat{\xi}) + \sqrt{n} \Ep  g(D_i, \theta,\widehat{\xi})\right)^\top \widehat{W}(\theta) \|_{+}^2  \\
&\geq \lambda_{\min} \| \Gn g(D_i, \theta, \widehat{\xi}) +  \sqrt{n} \Ep  g(D_i, \theta, \widehat{\xi}) \|_{+}^2\\
&\geq  \lambda_{\min}   \|  \sqrt{n} \Ep  g(D_i, \theta, \xi_0) \|_{+}^2  \bigg( \frac{  \| \Gn g(D_i, \theta, \xi_0) \|}{\|  \sqrt{n} \Ep  g(D_i, \theta, \xi_0) \|_{+}^2 } \\
&+ \frac{\| \Gn [g(D_i, \theta, \widehat{\xi}) - g(D_i, \theta, \xi_0)] + \sqrt{n} \Ep  [g(D_i, \theta, \widehat{\xi})-  g(D_i, \theta, \xi_0)] \|}{\|  \sqrt{n} \Ep  g(D_i, \theta, \xi_0) \|_{+}^2}  \\
&+ \frac{\|\sqrt{n} \Ep  g(D_i, \theta, \xi_0) \| }{ \|  \sqrt{n} \Ep  g(D_i, \theta, \xi_0) \|_{+}^2}   \bigg).
\end{align*}

Consider the set $$d_H(\theta,\Theta_I)\geq 6C_{p/4}/C_{\min} \sqrt{n}$$ and   $n_p$ large enough  so that $$\forall n \geq n_p, \quad  \sqrt{n} s_n< C_{p/4}, \quad r_n' \log (1/r_n') <C_{p/4}.$$ On this set, the following bounds apply:
\begin{align*}
  \| x\|_{+}&:= \inf_{\theta: d_H(\theta,\Theta_I)\geq 6C_{p/4}/C_{\min} \sqrt{n})}\|  \sqrt{n} \Ep  g(D_i, \theta, \xi_0) \|_{+} \geq C_{\min} ( d_H(\theta, \Theta_I) \wedge \delta_{\min}) \geq 6 C_{p/4} , \\
  \| y\| &:= \| \Gn [ g(D_i, \theta, \xi_0)] + \Gn [g(D_i, \theta, \widehat{\xi}) - g(D_i, \theta, \xi_0)] + \sqrt{n} \Ep  [g(D_i, \theta, \widehat{\xi})-  g(D_i, \theta, \xi_0)] \|\\
  &\leq C_{p/4} + 2C_{p/4} +  \leq 3 C_{p/4}.
 \end{align*}
Therefore, $\|x\|_{+} \geq 2\|y\|$. Plugging in $x$ and $y$  into the inequality below:
 \begin{align*}
 \frac{ \| x+y\|_{+} }{\| x\|_{+}}  \geq 1 - \frac{ \|y \|}{\|x \|_{+}}  \geq 1 - \frac{1}{2} =\frac{1}{2},  \\
 ( \frac{ \| x+y\|_{+} }{\| x\|_{+}} )^2 \geq \frac{1}{4}
  \end{align*}
implies that $ \frac{ \| x+y\|_{+} }{\| x\|_{+}}$ is greater than or equal to $\frac{1}{2}$ on the set $\mathcal{B}_n \cap \mathcal{W}_n \cap \mathcal{F}_n \cap \mathcal{G}_n$ for $n$ sufficiently large.
Setting $\kappa = C^2_{\min}/16, \delta = \delta_{\min}, \gamma = 1/2, \kappa_p:= (6C_{p/4}/C_{\min})^2$ and $n$ large enough implies:
\begin{align*}  \inf_{\xi \in \Xi_n} \Pr ( Q_n(\theta,\xi) \geq C^2_{\min}/16 [ d(\theta, \Theta_I) \wedge \delta_{\min}]^{2} \\
 \forall \theta: d_H(\theta,\Theta_I)\geq 6C_{p/4}/C_{\min} \sqrt{n} )\geq 1 - p. \end{align*}
 Therefore, Assumption \ref{ass:cons}(4) holds. Lemma \ref{lem:rateineq} verifies Assumption \ref{ass:noeffect}.
\end{proof}

\begin{proof}[Proof of Lemma \ref{lem:degeneracy}]
Step 1. Let us show that the degeneracy property and the choice $\widehat{c}'$ given in (\ref{eq:minc:deg}) suffice for the rate $O_{P}(n^{-1/2})$. Let  $\widehat{c}$ be as in (\ref{eq:minc}) and  $0 \leq \widehat{c}'  \leq \widehat{c} $ w.p. $\rightarrow 1$. Conditionally on the event $\mathcal{B}_n$ 
for any $\xi \in \Xi_n$ the following inclusion relation holds:
\begin{align*}
\inf_{\xi \in \Xi_n}	\Pr ( \Theta_n \subseteq {\cal C}_n (\widehat{c}', \xi) \subseteq {\cal C}_n(\widehat{c}, \xi)) \geq 1-p.
\end{align*}
By condition (2) of the degeneracy property, $d_H(\Theta_n ,\Theta_I) = O_{P} (n^{-1/2})$. It has been shown in Theorem \ref{thm:main:ineq:chap2} that $d_H( {\cal C}_n(\widehat{c}, \widehat{\xi}) = O_{P} (n^{-1/2})$. Therefore, $ d_H( {\cal C}_n(\widehat{c}', \widehat{\xi}) ,\Theta_I) = O_{P} (n^{-1/2})$ conditionally on $\mathcal{B}_n$. Since $\Pr(\mathcal{B}_n) \rightarrow 1$, the bound holds unconditionally.

Step 2.Suppose the conditions of Lemma \ref{lem:degeneracy} holds. Conditionally on $\mathcal{B}_n$  for  any $\xi \in \Xi_n$ the following bound holds uniformly on $\theta \in \Theta)I$:
\begin{align}
n Q_n(\theta, \xi) &\leq \lambda_{\max} \| \Gn g(W_i,\theta, \xi) + \sqrt{n} \E g(W_i,\theta, \xi)\|_{+}^2 \\
&\leq \lambda_{\max} \sum_{l =1}^L | \Gn g_l(W_i,\theta, \xi) + \sqrt{n} \E g_l(W_i,\theta, \xi) |_{+}^2 \\
&\leq \lambda_{\max} \sum_{l =1}^L  | o(1) + O_{P} (1) + \sqrt{n}  \E g_l(W_i,\theta, \xi_0) |_{+}^2 \\
&\leq \lambda_{\max} \sum_{l =1}^L  | o(1) + O_{P} (1) - \sqrt{n} C (d(\theta, \Theta \subset \Theta_I)  \wedge \delta)  |^2.
\end{align}
Conclude that $Q_n(\theta, \xi) = 0 \quad \forall \Theta_I^{-\epsilon_n}$ with $\epsilon_n :=2\sum_{l =1}^L L  |\En g(W_i,\theta,\xi_0) |/ \frac{1}{C}$ satisfies $\epsilon_n = O_{P} (1/\sqrt{n})$. Fix any $p >0$. Since $\epsilon_n $ does not depend on $\xi$, there exists $R_p$ and $N_p$ such that for all $n \geq N_p$
\begin{align*}
	\Pr ( \sqrt{n} d_H( \Theta_I^{-\epsilon_n}, \Theta_I)  \leq R_p) \geq 1-p.
\end{align*}

\end{proof}

\begin{proof}[Proof of Theorem \ref{thm:main:subs}]
Assumption \ref{ass:noeffect}[2] implies that $\widehat{c} = O_{P}(1)$ satisfies the conditions of Theorem \ref{lem:rate}. Let $\epsilon_n \asymp  (\frac{ \log^2 n }{n} )^{1/2}$ be a numerical sequence.  Fix a particular subsample $j$ and its  corresponding objective $Q_{j,b}$ for that subsample. Define $$\bar{C}_{j,b,\xi} := \sup_{\theta \in \Theta_I^{\epsilon_n}} b Q_{j,b}(\theta, \xi)$$ and $$\underbar{C}_{j,b,\xi} := \inf_{K \in K_n} \sup_{\theta \in K} b Q_{j,b}(\theta, \xi),$$
where $K_n = \{ \Theta_n: d_H(\Theta_n, \Theta) \leq \epsilon_n \}$ is closed set of all fixed(non-random) sets within  Hausdorff distance $ \epsilon_n$ of $\Theta_I$. Since $K_n$ is a closed set, there exists a set $\Theta^*_b$ where the infimum above is achieved:
 $$\underbar{C}_{j,b,\xi} := \inf_{K \in K_n} \sup_{\theta \in K} b Q_{j,b}(\theta, \xi)=  \sup_{\theta \in \Theta^*_b} b Q_{j,b}(\theta, \xi).$$
By Lemma \ref{lem:rate}(c), for any $p>0$ there exists $n^A_{p}$ such that  $\forall n \geq n^A_{p}$ $\widehat{c}<\epsilon_n$ holds with  probability at least $1-p$:
 \begin{align*}
\Pr (\mathcal{A}_{n}):&= \Pr( \widehat{c} < \epsilon_n) \geq 1-p.
\end{align*} Therefore, on the event $\mathcal{A}_{n}$ for any  $\xi \in \Xi_n^{\theta}$ the following inequality holds: $$\underbar{C}_{j,b,\xi }  \leq \widehat{{\cal C}}_{j,b} := \sup_{\theta \in {\cal C}_n(\widehat{c},\xi )}  b Q_{j,b}(\theta, \xi) \leq \bar{C}_{j,b,\xi}.$$
By Assumption \ref{ass:smallbiasvalue}, for $p>0$ there exists $n^B_{p}$ such that  $\forall n \geq n^B_{p}$ $$\Pr (\mathcal{B}_n):=\Pr(\widehat{\xi} \in \Xi_n^{\theta}) \geq 1-p.$$
Define the event 
$$D_{\epsilon, \xi}:=\{ \sup_{\theta \in \Theta_I^{\epsilon_n}} b |Q_{j,b}(\theta, \xi) -Q_{j,b} (\theta, \xi_0)| \leq \epsilon \}$$

By Assumption \ref{ass:fast}, for $p,\epsilon>0$ there exists $n^D_{p, \epsilon}$ such that  $\forall n \geq n^D_{p, \epsilon}$ and any $\xi \in \Xi_n$  
the event $D_{\epsilon, \xi}$ hold with probability at least $1-2p$ (i.e, $\Pr( D_{p,\widehat{\xi}} \cap \mathcal{B}_n ) \geq 1-2p$).

On the event $\mathcal{A}_{n} \cap \mathcal{B}_n \cap D_{p, \widehat{\xi}}$ for $n \geq  n^A_{p} \vee n^B_{p} \vee n^D_{p, \epsilon}$,

\begin{align*}
 \underbar{C}_{j,b,\xi_0} -\epsilon \leq^{i} \underbar{C}_{j,b,\widehat{\xi} }  \leq^{ii}  \widehat{{\cal C}}_{j,b} := \sup_{\theta \in {\cal C}_n(\widehat{c},\widehat{\xi} )}  b Q_{j,b}(\theta, \widehat{\xi}) \leq^{iii} \bar{C}_{j,b,\widehat{\xi}}  \leq^{iv} \bar{C}_{j,b,\xi_0} + \epsilon,
\end{align*}
where $i$ and $iv$ hold by definition of $D_{p, \widehat{\xi}}$ and the property of supremum; $ii$ and $iii$ hold by definition of the event $\mathcal{A}_{n}$ and $\mathcal{B}_n$.
Denote $$\widehat{G}_{\widehat{\xi}} (x):= \frac{1}{B_n} \sum_{j=1}^{B_n} 1_{ \{ \widehat{{\cal C}}_{j,b,\widehat{\xi}}  \leq x\}}, \quad  \bar{G} (x):=\frac{1}{B_n} \sum_{j=1}^{B_n} 1_{ \{ \underbar{C}_{j,b,\xi_0}  \leq x\}}, \quad \underbar{G} (x):=\frac{1}{B_n} \sum_{j=1}^{B_n} 1_{ \{ \bar{C}_{j,b,\xi_0}  \leq x\}}.  $$ Since $$ \widehat{{\cal C}}_{j,b,\widehat{\xi}}  \leq \bar{C}_{j,b,\xi_0} + \epsilon $$ the event $\bar{C}_{j,b,\xi_0} \leq x-\epsilon$ implies  $\widehat{{\cal C}}_{j,b}\leq x.$ Therefore, $ \frac{1}{B_n} \sum_{j=1}^{B_n} 1_{ \{\widehat{{\cal C}}_{j,b,\widehat{\xi}}  \leq x\}} \geq \frac{1}{B_n} \sum_{j=1}^{B_n} 1_{\{ \bar{C}_{j,b,\xi_0} \leq x-\epsilon\}}$, and using the notation above, I get:
$$ \underbar{G}(x- \epsilon) \leq \widehat{G}_{\widehat{\xi}} (x). $$ Similar argument gives: $$ \underbar{G}(x- \epsilon) \leq \widehat{G}_{\widehat{\xi}} (x) \leq \bar{G} (x+\epsilon). $$ Step 2 shows that
continuity of the c.d.f of $ \Pr({\cal C} \leq x)$ (Assumption \ref{ass:c4}) implies that $$ \underbar{G} (x-\epsilon) \rightarrow_p \Pr({\cal C} \leq x-\epsilon) \rightarrow  \Pr({\cal C} \leq x), b \rightarrow \infty $$ and $$ \bar{G} (x+\epsilon) \rightarrow_p \Pr({\cal C} \leq x+\epsilon) \rightarrow  \Pr({\cal C} \leq x), b \rightarrow \infty . $$

Step 2. Consider the  function $$h(W_1,W_2,..,W_n):= \frac{1}{B_n} \sum_{j=1}^{B_n} 1_{ \{ \bar{C}_{j,b,\xi_0}  \leq x - \epsilon\}}. $$ Let us show that $h(\cdot)$ has bounded differences: replacement of observation $W_k$ by $W_k'$ changes at most one subsample $j$ and results in at most a  $\frac{1}{B_n}$ change in the function itself.  $$ | h(W_1,W_k,..,W_n) - h(W_1,W_k',..,W_n) | \leq \frac{1}{B_n}.$$
McDiarmid's inequality implies:
\begin{align*}
\sup_{\epsilon \in \mathcal{R}} \Pr (| \underbar{G} (x-\epsilon)  &- \Ep \underbar{G} (x-\epsilon)  | > t ) \leq 2 \exp^{-2t^2 B_n^2/ n  } := \delta \\
 \underbar{G} (x-\epsilon) &= \Ep  \underbar{G} (x- \epsilon)    +  O_{P} ( (\log \frac{2}{\delta})^{1/2} \frac{\sqrt{n}}{B_n} )
\end{align*}
uniformly over $\epsilon>0$. Since the subsamples $j_1\neq j_2$ are i.i.d, $$   \Ep  \underbar{G} (x-\epsilon)   = \Pr( \bar{C}_{j,b,\xi_0}  \leq x-\epsilon) .$$ 
Step 3. By Assumption \ref{ass:c5},
\begin{align}
  &| \Pr( \bar{C}_{j,b,\xi_0}  \leq x-\epsilon) -  \Pr({\cal C} \leq x)| \leq  \sup_{x \in \mathcal{R}}  | \Pr( \bar{C}_{j,b,\xi_0}  \leq x ) -  \Pr({\cal C} \leq x)| \tag{ Assumption \ref{ass:c5}}  \\
  &+  | \Pr({\cal C} \leq x-\epsilon) - \Pr({\cal C} \leq x)| \tag{Continuity of $C_{\nu_0}$} \\
  &= o(1)
  \end{align}
 as $n \rightarrow \infty, b \rightarrow \infty.$ Since the bound above holds for any $\epsilon>0$, the statement is proved.
To conclude, I have shown that $\bar{G} (x+\epsilon) = P({\cal C} \leq c) + o_{P} (1)$. Similarly, I can show that $\underbar{G} (x-\epsilon) = P({\cal C} \leq c) + o_{P} (1)$. Since $P({\cal C} \leq c)$ is continuous in $c$, the $\alpha$- quantile of $\widehat{G}_{\widehat{\xi}} (x)$ converges to the $\alpha$-quantile of ${\cal C}$.
\end{proof}

\subsection{Proof of Auxiliary Lemmas}
\label{subsec:proof:aux}
\begin{proof}[Proof of Lemma \ref{lem:rate}]
The proof relies on the following basic inequalities. For any set $\Theta$ and two functions $p(\theta), q(\theta)$ the following holds: \begin{align}
\label{eq:basic}
\sup_{\theta \in \Theta} p(\theta) \leq \sup_{\theta \in \Theta} q(\theta) + \sup_{\theta \in \Theta} (p(\theta) - q(\theta))_{+} \\
\inf_{\theta \in \Theta} p(\theta) \geq \inf_{\theta \in \Theta} q(\theta) - \sup_{\theta \in \Theta} (q(\theta) - p(\theta))_{+}.
\label{eq:basicinf}
 \end{align}
Consider a sequence of events  $ \mathcal{B}_{n} := \{ \widehat{\xi}(\cdot) \in \Xi_n \}$ whose probability approaches one: $$\Pr ( \mathcal{B}_{n}) = 1-o(1).$$ 
Step 1.
On an event $ \mathcal{B}_{n}$ $$\Pr(\sup_{\theta \in \Theta_I} n Q_n (\theta, \widehat{\xi}) \leq \widehat{c} | \mathcal{B}_{n}) \geq \inf_{\xi \in \Xi_n} \Pr(\sup_{\theta \in \Theta_I} n Q_n (\theta, \xi) \leq \widehat{c} ) = 1-o(1),$$ which implies
\begin{align*}  \Pr(\Theta_I \subset \widehat{\Theta}_I | \mathcal{B}_{n}) &= 1-o(1).
 \end{align*}
Since $\Pr(\mathcal{B}_n) = 1-o(1)$, with probability approaching one $\Pr(\Theta_I \subset \widehat{\Theta}_I)$ holds.

Step 2. Proof of convergence without guaranteed rate.  Fix  an $\epsilon>0$. Let us show that  $$\Pr (d_H(\widehat{\Theta}_I, \Theta_I) \leq \epsilon ) \rightarrow 1.$$ By Assumption \ref{ass:cons} (Equation \eqref{eq:cons1}), there exists $\delta(\epsilon)$ such that $$ \inf_{\Theta \setminus {\Theta_I}^{\epsilon}} Q(\theta, \xi_0) \geq \delta(\epsilon)>0 .$$ To see that $\Pr(\widehat{\Theta}_I \subset {\Theta_I}^{\epsilon}) = 1-o(1)$, recognize that on the event $\mathcal{B}_{n}$ conditionally on the subsample $J_1$:
\begin{align*}
	&\sup_{\widehat{\Theta}_I} Q (\theta, \xi_0) \leq  \sup_{\widehat{\Theta}_I} Q_n (\theta, \xi_0) + \sup_{\widehat{\Theta}_I} (Q (\theta, \xi_0) -Q_n (\theta, \xi_0)  )_{+} \tag{Eq. \eqref{eq:basic}} \\
	&\leq  \sup_{\widehat{\Theta}_I} Q_n (\theta, \xi_0) + O_{P} (1/\sqrt{n})  \tag{Eq. \eqref{eq:cons3}}  \\
	&\leq  \sup_{\widehat{\Theta}_I} Q_n (\theta, \widehat{\xi}(\theta) ) +   \sup_{\widehat{\Theta}_I} (Q_n (\theta, \xi_0) -Q_n (\theta, \widehat{\xi})  )_{+}    + O_{P} (1/\sqrt{n})\tag{Eq. \eqref{eq:basic}}  \\
	&\leq \widehat{c} / n   + O_{P} (1/\sqrt{n}+1/\sqrt{n})  = o_{P} (1) \tag{Assumption \ref{ass:noeffect} and the choice of $\widehat{c}$ }
\end{align*}
Since $ \sup_{\widehat{\Theta}_I} Q (\theta, \xi_0) < \frac{\epsilon}{2}  \Rightarrow \widehat{\Theta}_I \subset {\Theta_I}^{\epsilon}$, $\Pr(\widehat{\Theta}_I \subset {\Theta_I}^{\epsilon}) = 1-o(1)$.

Step 3. Proof of convergence at rate $\epsilon_n:= (\frac{d_{p/3} \kappa \vee \widehat{c}}{n \kappa})^{1/2}.$ Fix a  probability level $p \in (0,1)$ and let the constants $\kappa,\delta$ be  as specified in Assumption \ref{ass:cons}(Equation \eqref{eq:cons4}).
I have to show:
\begin{align*}
\forall p>0 \quad \exists d_p, n_p: \quad \forall n \geq n_p  \quad \Pr ( d_H( \widehat{\Theta}_I, \Theta_I) \leq \epsilon_n  ) \geq 1-p.
\end{align*}
Since $\widehat{c}/n \rightarrow_p 0$, there exists $n^{A}_{p/3}$ such that for a sufficiently large  $n: n \geq n^{A}_{p/3}$,
$$\Pr(\mathcal{A}_n):=\Pr( \widehat{c}/(n \kappa) < (\delta/2)^{2}) = 1-p/3. $$ By the definition of $\widehat{\xi}(\theta)$, for a sufficiently large  $n: n \geq n^{B}_{p/3}$,
$$\Pr( \mathcal{B}_{n} ) := \Pr (\widehat{\xi}(\cdot) \in \Xi_n) \geq 1-p/3.$$
Define a set $$\mathcal{D}_{\widehat{\xi}} :=\{\inf_{\theta: d_H(\theta,\Theta_I) \geq \epsilon_n }  n Q_n (\theta, \widehat{\xi}) \geq n \kappa (\epsilon_n \wedge \delta)^{2} \}$$
Since $\epsilon_n \geq (d_{p/3}/n)^{1/2}$ holds absolutely surely by the choice of $\epsilon_n$, Assumption \ref{ass:cons}[4] implies that for $n \geq n_{p/3}$:
$$\Pr( \mathcal{D}_{\widehat{\xi}} |  \mathcal{B}_{n} ) \geq \inf_{\xi \in \Xi_n} \Pr (\inf_{\theta: d_H(\theta,\Theta_I) \geq \epsilon_n }  n Q_n (\theta, \xi) \geq n \kappa (\epsilon_n \wedge \delta)^{2}) \geq 1-p/3 .$$
Therefore, for a sufficiently large  $n: n \geq n^{A}_{p/3} $,
\begin{align}
\label{eq:2}
 \inf_{\theta: d_H(\theta,\Theta_I) \geq \epsilon_n }  n Q_n (\theta, \widehat{\xi}) &\geq n \kappa (\epsilon_n \wedge \delta)^{2} \Rightarrow \\
 \inf_{\theta: d_H(\theta,\Theta_I) \geq \epsilon_n }  n Q_n (\theta, \widehat{\xi}) &\geq n  \kappa \epsilon_n^{2} =  d_p \kappa \vee \widehat{c} \geq  \widehat{c} \nonumber.
 \end{align}
 Therefore, for a sufficiently large $n: n \geq n_{p/3} \vee n^{A}_{p/3} \vee n^{B}_{p/3} $,
\begin{align} \label{eq:1} \Pr (  \inf_{\theta: d_H(\theta,\Theta_I) \geq \epsilon_n }  n Q_n (\theta, \widehat{\xi})  \geq \kappa (\epsilon_n \wedge \delta)^{2}  |\mathcal{B}_n) \\
 \geq \inf_{\xi \in \Xi_n} \Pr (  \inf_{\theta: d_H(\theta,\Theta_I) \geq \epsilon_n }  n Q_n (\theta, \xi)  \geq \kappa (\epsilon_n \wedge \delta)^{2}  ) \nonumber.  \end{align} Combining \eqref{eq:1} and  \eqref{eq:2} gives

\begin{align*}
\Pr (  \inf_{\theta: d_H(\theta,\Theta_I) \geq \epsilon_n }  n Q_n (\theta, \widehat{\xi})  \geq  \kappa (\epsilon_n \wedge \delta)^{2} \cap   \mathcal{B}_n \cap \mathcal{A}_n   )   \\
\geq 1- \Pr ( (\mathcal{B}_n)^c) + \Pr ( (\mathcal{A}_n)^c) + \Pr ( (\mathcal{D}_{\widehat{\xi}})^c)\\
\geq 1-(p/3+p/3+p/3).
  \end{align*}
Since $\Pr ( \sup_{\widehat{\Theta}_I}  n Q_n (\theta, \widehat{\xi})  \leq \widehat{c} ) = 1-o(1)$, this implies that $\Pr(\widehat{\Theta}_I \subset \Theta^{\epsilon_n}_I) = 1-o(1).$ By Step 1, I conclude that  $d_H(\widehat{\Theta}_I, \Theta_I)  = O_{P}(\epsilon_n)$.

 Step 4. That $\widehat{c} = O_{P}(1)$ satisfies the conditions of Lemma \ref{lem:rate}, follows from
\begin{align*}
\sup_{\theta \in \Theta_I} &n Q_n(\theta, \xi) \leq \sup_{\theta \in \Theta_I} n Q_n(\theta, \xi_0) + \sup_{\theta \in \Theta_I} n |Q_n(\theta, \xi) - Q_n(\theta, \xi_0) |  \\
&\leq O_{P}(1) + \sup_{\theta \in \Theta_I} n |Q_n(\theta, \xi) - Q_n(\theta, \xi_0) | \tag{Assumption \ref{ass:cons}[2]} \\
&\leq O_{P}(1) + o_{P}(1) \tag{Assumption \ref{ass:noeffect}[2]}
\end{align*}
for any $\xi \in \Xi_n$.

\end{proof}

\begin{proof}[Proof of Lemma \ref{lem:concentrate}]
Conditionally on the auxiliary sample, $\widehat{\xi}(\cdot)$ can be treated as fixed, and w.p. approaching one, $\widehat{\xi}(\cdot) \in \Xi_n$. Consider the function class: $$\mathcal{F}_2 = \{ g_l(D_i, \theta, \widehat{\xi}) -  g_l(D_i, \theta,\xi_0), l=1,2,..,L, \theta \in \Theta \} \subset  \mathcal{F}_{ \widehat{\xi}} -  \mathcal{F}_{\xi_0}.  $$ Let a function $F_2 := F_{ \widehat{\xi}} + F_{\xi_0} $ be the envelope function for the class $\mathcal{F}_2 $. This function satisfies the envelope requirements of Lemma 6.2   of \cite{doubleml2016}  since  $\|F_2 \|_{P,q} \leq \|F_{\widehat{\xi}} \|_{P,q} + \| F_{\xi_0} \|_{P,q} \leq 2 C_1$ and $$ \log \sup_{\tilde{Q}} N( \epsilon \| F_2 \|_{Q,2} , \mathcal{F}_2, \| \cdot \|_{\tilde{Q},2} ) \leq 2v \log (2a /\epsilon). $$ On the event $\mathcal{B}_n$,  \begin{align*}  \sup_{\theta \in \Theta}  \Ep [g_j(D_i, \theta, \widehat{\xi}) -  g_j(D_i, \theta,\xi_0)]^2 &\leq \sup_{\theta  \in \Theta, \xi \in \Xi_n} \Ep \|  [g_j(D_i, \theta, \xi(\theta)) -  g_j(D_i, \theta,\xi_0)] \|^2\\
&\leq (r_n')^2 \end{align*}
The application of Lemma 6.2  of \cite{doubleml2016} conditionally on the auxiliary sample, with the function class $\mathcal{F}_2$, envelope $F_2$, and $\sigma^2:=\sqrt{n}'^2$ yields:
\begin{align}
\sup_{f \in \mathcal{F}} | \Gn  [g(D_i, \theta, \widehat{\xi} ) - g(D_i, \theta, \xi_0) ]|  \nonumber \\
\lesssim_{P} r_n' \log^{1/2} (1/r_n') + n^{-1/2+1/s} \log n = o_{P} (1) \label{eq:moment}
\end{align}
% By Lemma \ref{lem:condconv}, conditional convergence in Equation \ref{eq:moment} implies unconditional convergence in Equation \ref{eq:moment2}.
\end{proof}

\begin{proof}[Proof of Lemma \ref{lem:rateineq}] 
 
Let the events $\mathcal{W}_n, \mathcal{B}_n, \mathcal{F}_n, \mathcal{G}_n$ be as defined at in the Proof of Theorem \ref{thm:main:ineq:chap2}.
The application of Lemma \ref{lem:basicineq} on the event $\mathcal{W}_n \cap \mathcal{B}_n \cap  \mathcal{F}_n \cap \mathcal{G}_n$ gives:
\begin{align*}
| Q_n(\theta, \xi) -Q_n(\theta, \xi_0)| &= | \|\underbrace{ \En  g(D_i, \theta, \widehat{\xi})\widehat{W}(\theta)}_{x} \|_{+}^2  - \| \underbrace{\En  g(D_i, \theta, \xi_0)\widehat{W}(\theta)}_{y} \|_{+}^2  | \\
&\leq \|x-y \|_{+}^2 + 2\|x-y\|_{+} \|y\|_{+}.
\end{align*}
The terms $\|x-y \|_{+}$ and $\|y\|_{+}$ admit the following bound:
\begin{align*}
    \|x-y \|_{+} &= \| \En [ g(D_i, \theta, \xi) - g(D_i, \theta, \xi_0)] \widehat{W}(\theta) \|\\
    &\leq \lambda_{\max} \| \Ep  [ g(D_i, \theta, \xi) - g(D_i, \theta, \xi_0)] + \Gn [ g(D_i, \theta, \xi) - g(D_i, \theta, \xi_0)]/\sqrt{n}   \|\\
    &\leq \lambda_{\max} (s_n + (r_n' \log (1/r_n')  + n^{-1/2+1/s}) ), \\
    \|y\|_{+} &= \| \Gn  g(D_i, \theta, \xi_0)/\sqrt{n} +  \Ep g(D_i, \xi_0)\|_{+}.
\end{align*}
Therefore, Assumption \ref{ass:noeffect}(a) holds, and $$r_p = (2\lambda_{\max}C_{p/4})^2 + 2\lambda_{\max}C_{p/4}2(  C_{p/4}+\sup_{\theta \in \Theta} \| \Ep  g(D_i, \theta, \xi_0)  \| )$$ satisfies
\begin{align*}
\inf_{\xi \in \Xi_n} \Pr ( \sqrt{n} \sup_{\Theta}   |Q_n(\theta, \xi) -Q_n(\theta, \xi_0) | \leq r_p) \geq 1-p .
\end{align*}
and for $n$ large enough $\sup_{\theta \in \Theta_I} \| C_{p/4} +\sqrt{n} \Ep g(D_i,\theta,\xi_0) \|_{+} = 0 $, $r'_p = (2\lambda_{\max}C_{p/4/})^2 $
\begin{align*}
\inf_{\xi \in \Xi_n} \Pr ( n \sup_{\Theta_I}   |Q_n(\theta, \xi) -Q_n(\theta, \xi_0) | \leq r_p') \geq 1-p .
\end{align*}
If Assumption \ref{ass:dominance}(2) holds,
\begin{align*}
    \sup_{\Theta_I^{\epsilon}}   \quad &n |Q_n(\theta, \xi) -Q_n(\theta, \xi_0) | \leq n \|x-y \|_{+}^2 + 2n\|x-y\|_{+} \|y\|_{+}\\
    &\sqrt{n} \|x-y \|_{+} = \| \sqrt{n} \En [ g(D_i, \theta, \xi) - g(D_i, \theta, \xi_0)] \widehat{W}(\theta) \|\\
    &\leq \lambda_{\max} \| \sqrt{n}\Ep  [ g(D_i, \theta, \xi) - g(D_i, \theta, \xi_0)] + \Gn [ g(D_i, \theta, \xi) - g(D_i, \theta, \xi_0)]   \|\\
    &\leq \lambda_{\max} (\sqrt{n} s_n + r_n'\log(1/r_n') + n^{-1/2+1/s}) =o(1). \\
    \sqrt{n}\|y\|_{+} &= \| \Gn  g(D_i, \theta, \xi_0)+  \sqrt{n} \Ep g(D_i, \xi_0) \|_{+}  \\
    \sup_{\theta \in \Theta_I^{\epsilon_n}} \sqrt{n}\|y\|_{+} &\leq C_{p/4} + \sqrt{n} \epsilon_n \leq C_{p/4} + O_{P} (1).
\end{align*}
Therefore, for any $\epsilon,p>0$ for $n$ sufficiently large
\begin{align*}
\inf_{\xi \in \Xi_n} \Pr (\sup_{\Theta_I^{\epsilon_n}}   n |Q_n(\theta, \xi) -Q_n(\theta, \xi_0) | \leq \epsilon) \geq 1-p.
\end{align*}

\end{proof}

\bibliographystyle{apalike}
\bibliography{my_bibtex}
\end{document}